\documentclass[letterpaper,11pt]{article}
\usepackage[utf8]{inputenc}
\usepackage[namelimits]{amsmath} 
\usepackage{forest}
\usepackage{fullpage}
\usepackage{color}
\usepackage{amsfonts}
\usepackage[left=1.0in,right=1.0in,top=1.0in,bottom=1.0in]{geometry}
\definecolor{Blue}{rgb}{0.1,0.1,0.8}
\usepackage{hyperref}
\hypersetup{
    linktocpage=true,
    colorlinks=true,				
    linkcolor=Blue,				
    citecolor=Blue,				
    urlcolor=Blue,			
}
\usepackage{makecell}
\usepackage{tablefootnote}
\usepackage{graphicx}  
\usepackage{hyperref}
\usepackage[ruled]{algorithm2e}
\usepackage{mathrsfs}
\usepackage[english]{babel}
\usepackage{amssymb}
\usepackage{multirow}
\usepackage{soul}
\usepackage{enumitem}

\usepackage{amsthm}
\usepackage{amsmath}
\usepackage{verbatim}
\usepackage{amssymb,tikz}
\usepackage{blindtext}
\usepackage{booktabs}
\newtheorem{theorem}{Theorem}[section]

\newtheorem{lemma}[theorem]{Lemma}
\theoremstyle{definition}

\newtheorem{definition}{Definition}
\newcommand{\E}{\mathsf{E}}

\newcommand{\bucket}{b}

\title{Universal Private Estimators}
\date{}

\author{
Wei Dong\thanks{Computer Science Engineering, Hong Kong University of Science and Technology. {\tt wdongac@cse.ust.hk}. }
\and
Ke Yi\thanks{Computer Science Engineering, Hong Kong University of Science and Technology. {\tt yike@cse.ust.hk}. }
}

\begin{document}
\maketitle 

\begin{abstract}
We present \textit{universal} estimators for the statistical mean, variance, and scale (in particular, the interquartile range) under pure differential privacy. These estimators are universal in the sense that they work on an arbitrary, unknown continuous distribution $\mathcal{P}$ over $\mathbb{R}$, while yielding strong utility guarantees except for ill-behaved $\mathcal{P}$.  For certain distribution families like Gaussians or heavy-tailed distributions, we show that our universal estimators match or improve existing estimators, which are often specifically designed for the given family and under \textit{a priori} boundedness assumptions on the mean and variance of $\mathcal{P}$. 
This is the first time these boundedness assumptions are removed under pure differential privacy. 
The main technical tools in our development are instance-optimal empirical estimators for the mean and quantiles over the unbounded integer domain, which can be of independent interest.
\end{abstract}

\section{Introduction}
\label{sec:intro}

Parameter estimation is a central problem in statistics, data mining, and machine learning. Let $\mathcal{P}$ be a continuous probability distribution over $\mathbb{R}$ with density function (pdf) $f(x)$, and let $F(x)$ be its  cumulative distribution function (CDF). We consider the following three fundamental parameters, mean, variance, and IQR:
\begin{align*}
\mu_{\mathcal{P}} &= \int_{-\infty}^{\infty} x f(x) \, \mathrm{d} x, 
\\
\sigma^2_{\mathcal{P}} &= \int_{-\infty}^{\infty} (x-\mu)^2 f(x) \, \mathrm{d} x, 
\\
\mathrm{IQR}_{\mathcal{P}} &= F^{-1}(3/4)-F^{-1}(1/4).
\end{align*}
Note that the \textit{interquartile range (IQR)} is a widely used parameter for the \textit{scale} of $\mathcal{P}$, but the particular choices of $1/4$ and $3/4$ are not very important: changing them to other constants does not affect our results (for both error bound and the requirement of $n$) asymptotically.  For simplicity, we omit the subscript $\mathcal{P}$ when there is no confusion.

Given an i.i.d.\ sample $D=(X_1,\dots, X_n)$ drawn from $\mathcal{P}^n$, the standard estimators for these parameters are (we reorder $D$ such that $X_1\le \cdots \le X_n$):
\begin{align*}
\mu(D) &= {1\over n}\sum X_i, 
\\
\sigma^2(D) &= {1\over n}\sum (X_i - \mu(D))^2, 
\\
\mathrm{IQR}(D) &= X_{3n/4} - X_{n/4},
\end{align*}
which are often called the \textit{sample} or \textit{empirical} mean, variance, and IQR.  They all converge to the true parameter respectively at a rate of $O(1/\sqrt{n})$, and the difference between the empirical parameter and the statistical parameter is referred to as the \textit{sampling error}. Importantly, all these estimators are \textit{universal}, namely, they work on an arbitrary, unknown $\mathcal{P}$.  The $O(1/\sqrt{n})$ convergence rate is optimal for many families of distributions, but not all.  For instance, the mid-range estimator $(X_1+X_n)/2$ is a better estimator of $\mu$ for uniform distributions with a convergence rate of $O(1/n)$.  However, such distribution-specific estimators are less used in practice as we usually do not know which family $\mathcal{P}$ is chosen from, and they may fail miserably when the distributional assumption does not hold (e.g., the mid-range estimator is a very bad estimator of the Gaussian mean).

In this paper, we design universal estimators under \textit{differential privacy (DP)}~\cite{dwork2014algorithmic}. A randomized mechanism $\mathcal{M}:\mathcal{X}^n\rightarrow \mathcal{Y}$ satisfies $(\varepsilon,\delta)$-DP if for any two neighboring datasets $D \sim D'$ (i.e., $D$ and $D'$ differ by one record), and any $\mathcal{S}\subseteq \mathcal{Y}$, 
\begin{equation}\label{eq:DP}
    \Pr[\mathcal{M}(D)\in \mathcal{S}]\leq e^\varepsilon\cdot \Pr[\mathcal{M}(D')\in \mathcal{S}]+\delta,
\end{equation}
for some privacy parameters $0<\varepsilon <1, 0\le \delta <1/n^{\omega(1)}$.  For statistical estimation problems, the high-privacy regime (e.g., $\varepsilon < 1/\sqrt{n}$) is more interesting; otherwise, the error would be dominated by the sampling error for many distributions (i.e., privacy is free). This is because the privacy error is $\tilde{O}(1/(\epsilon n))$ while the sampling error is $\tilde{O}(1/\sqrt{n})$.
The case $\delta=0$ is often called \textit{pure DP}, abbreviated as $\varepsilon$-DP.  It is preferable than the $\delta>0$ case, since $\delta$ corresponds to the probability of catastrophic privacy breaches.  However, there are strong separation results showing that for certain problems, $\varepsilon$-DP is strictly harder to achieve than $(\varepsilon, \delta)$-DP~\cite{hardt2010geometry,de2012lower,beimel2013private,bun2015differentially,vadhan2017complexity}. Note that, when designing a private estimator, the DP guarantee should hold for any two neighboring datasets $D, D'$, and \eqref{eq:DP} is only over the internal randomness of $\mathcal{M}$. When analyzing its utility, however, the randomness in both $D$ and $\mathcal{M}$ is taken into consideration. 

\begin{table}[ht]
\centering
\label{tab:compare}
\begin{tabular}{c|c|c|c}
\hline
& $\mu$ & $\sigma^2$ & $\mathrm{IQR}$
\\
\hline
$\varepsilon$-DP & \makecell{A1, A2, A3~\cite{smith2011privacy}\\ A1, A2, A3~\cite{karwa2018finite} \\ A1, A2, A3~\cite{KamathLSU19} \\ A1, A2~\cite{BunS19} \\ A1, A2, A3~\cite{bun2019private}\\A1, A2, A3~\cite{biswas2020coinpress} \\ A1, A2~\cite{kamath2020private}\\A1, A2~\cite{hopkins2021efficient}} & \makecell{A1, A2, A3~\cite{karwa2018finite} \\ A2, A3~\cite{KamathLSU19} \\ A1, A2, A3~\cite{bun2019private}\\ A2, A3~\cite{biswas2020coinpress}} & \makecell{None}
\\
\hline
$(\varepsilon,\delta)$-DP & \makecell{A3~\cite{karwa2018finite} \\ A1, A2~\cite{BunS19} \\A1, A2, A3~\cite{KamathLSU19}\\
A2, A3~\cite{bun2019private}\\A1, A2, A3~\cite{cai2019cost}\\ A1, A2, A3~\cite{biswas2020coinpress}\\ 
A3~\cite{aden2021sample}\\ A1, A2, A3~\cite{huang2021instance}\\A3~\cite{kamath2021private}\\A3~\cite{brown2021covariance}\\A3~\cite{liu2021differential}\\A3~\cite{ashtiani2021private}\\A3~\cite{kothari2021private}} & \makecell{A3~\cite{karwa2018finite} \\ A2, A3~\cite{KamathLSU19} \\ A2, A3~\cite{biswas2020coinpress}\\ 
A3~\cite{aden2021sample}\\A3~\cite{kamath2021private}\\A3~\cite{liu2021differential}\\A3~\cite{ashtiani2021private}\\A3~\cite{kothari2021private}} & \cite{dwork2009differential}
\\
\hline
\end{tabular}
\caption{Summary of existing private estimators\tablefootnote{Some estimators satisfy  CDP~\cite{bun2016concentrated}, which is between $\varepsilon$-DP and $(\varepsilon,\delta)$-DP. We classify them into $(\varepsilon,\delta)$-DP.} and their assumptions.}
\end{table}

In the past several years, quite a number of private estimators have been proposed in the literature as summarized in Table~\ref{tab:compare}.  With the exception of the IQR estimator of  \cite{dwork2009differential}, which only satisfies $(\varepsilon, \delta)$-DP, none of them is universal.  They all rely on the following three assumptions or a subset of them:
\begin{itemize}
  \item[A1.] a predefined range for the mean, i.e., $\mu\in[-R,R]$;
  \item[A2.] a predefined range for the variance, i.e.,  $\sigma^2\in[\sigma^2_{\mathrm{min}},\sigma^2_{\mathrm{max}}]$, as well as ranges for the higher moments if applicable;
  \item[A3.] $\mathcal{P}$ is chosen from a specific family of distributions such as Gaussian.
\end{itemize}

In particular, their reliance on A1/A2 is both algorithmic and analytical, i.e., these estimators need $R, \sigma_{\min}, \sigma_{\max}$ together with $D$ as the input, and the utility guarantees also depend on these \textit{a priori} bounds.  The reliance on A3 is only analytical; when we write A3 in Table \ref{tab:compare}, the corresponding estimator does not offer utility guarantees when $\mathcal{P}$ is chosen outside the specified family.\footnote{\cite{liu2021differential} can handle different distribution families but need to manually adjust the algorithm based on the distribution family.}

In this paper, we design universal private estimators under pure DP for $\mu, \sigma^2$, and $\mathrm{IQR}$ without these assumptions while achieving the same or better utilities.  As shown in Table~\ref{tab:compare}, this is the first time A1/A2 have been removed under pure-DP.
Under $(\varepsilon,\delta)$-DP, a number of prior works \cite{karwa2018finite,kamath2021private,bun2019private,aden2021sample,brown2021covariance,liu2021differential,ashtiani2021private,kothari2021private} show how A1/A2 can be removed, using \textit{stability} based techniques \cite{dwork2006calibrating,thakurta2013differentially,bun2016Simultaneou,vadhan2017complexity,bun2018composable}, the \textit{propose-test-release} framework~\cite{dwork2009differential}, or the \textit{truncated distribution}~\cite{chan2019foundations}. 
However, these techniques fundamentally do not work under pure DP.  More precisely, for the stability based techniques and the truncated distribution, even the output domain is different for neighboring datasets. The propose-test-release framework by nature must have a small probability that the privacy is breached, thus can only achieve $(\epsilon,\delta)$-DP.

As a necessary consequence, the utility guarantees of our estimators depend on the parameters of $\mathcal{P}$ to be estimated, namely, they are specific to the unknown $\mathcal{P}$.  As we shall see, our instance-specific results not only remove all boundedness assumptions, but also offer better utilities on most $\mathcal{P}$'s, compared to existing estimators that aim at optimizing the worst case (i.e., minimax bounds).  Finally, all our estimators can be implemented efficiently in $O(n \log n)$ time.

\subsection{Our Results and Comparison to Prior Work}
\label{sec:our_results}

Our general approach is as follows.  We first study the empirical problem, in particular, estimating the empirical mean $\mu(D)$ and the $\tau$-th quantile $X_\tau$ for any given $D$.  These empirical estimators only work over discrete domains, but we can apply them in the statistical setting by appropriately discretizing $\mathbb{R}$.  To remove A1/A2, we make our empirical estimators work over an infinite but discrete domain, namely, $\mathbb{Z}$.  To remove A3, we show that the errors achieved by our empirical estimators are instance-optimal, hence adaptive to an arbitrary $\mathcal{P}$ when applied in the statistical setting. Although our main motivation is in the statistical setting, the instance-optimality of our empirical estimators is of independent interest. 

\subsubsection{Empirical Estimators}
\label{sec:our_results_empirical}

Let $D=\{X_{1},\dots, X_{n}\}$ be a multiset drawn from $\mathbb{Z}$, and assume $X_1\le \cdots \le X_n$.  Estimating $\mu(D)$ and $X_{\tau}$ under DP has been studied previously, but existing algorithms either do not provide utility guarantees~\cite{mcmahan2017learning,amin2019bounding, andrew2019differentially, pichapati2019adaclip} or only work over a finite domain $[N] =\{0,1,\dots, N\}$ \cite{nissim2007smooth,asi2020instance,huang2021instance}. 

To reduce the domain from $\mathbb{Z}$ to a finite one, the natural idea is to use the \textit{empirical range} $\mathcal{R}(D) = [X_1,X_n]$ as the domain. However, doing so violates DP, and we must use a privatized $\tilde{\mathcal{R}}(D)$.  A good $\tilde{\mathcal{R}}(D)$ should be close to $\mathcal{R}(D)$ in both location and scale.  We thus approach the problem in two steps.  First, we obtain a privatized \textit{radius} of $D$, which is defined as $\mathrm{rad}(D) = \max_{i}|X_{i}|$. We show that our privatized radius is not too much larger than $\mathrm{rad}(D)$ while covering all but $O\left(\log\log(\mathrm{rad}(D))/\varepsilon\right)$\footnote{In this paper, we use $e$ as the base of $\log$ and define $\log(x) = 1$ for any $x\leq e$, unless stated otherwise.} elements in $D$:

\begin{theorem}[Theorem~\ref{th:err_infinite_radius}, informal]
\label{th:err_infinite_radius_intro}
There exists an $\varepsilon$-DP mechanism such that for any $D\in \mathbb{Z}^n$, it returns a  $\widetilde{\mathrm{rad}}(D)$ such that $\widetilde{\mathrm{rad}}(D)\leq 2\cdot \mathrm{rad}(D)$ and $\left|D\cap\overline{\left[-\widetilde{\mathrm{rad}}(D),\widetilde{\mathrm{rad}}(D)\right]} \right| = O\left(\frac{1}{\varepsilon}\log\log\left(\mathrm{rad}(D)\right)\right)$.\footnote{All results stated in Section~\ref{sec:intro} hold with constant success probability.}
\end{theorem}

In the second step, we try to find a rough location of $\mathcal{R}(D)$.  As we have bounded most elements into $\left[-\widetilde{\mathrm{rad}}(D),\widetilde{\mathrm{rad}}(D)\right]$, this can be done by using a finite-domain private median (Section~\ref{sec:INV}).  Then we shift $D$ to the median and invoke again our private radius estimator.  This results in a privatized $\tilde{\mathcal{R}}(D)$, whose width is not too much larger than the actual width $\gamma(D)=X_n - X_1$:

\begin{theorem}[Theorem~\ref{th:err_infinite_range}, informal]
\label{th:err_infinite_range_intro}
There exists an $\varepsilon$-DP mechanism such that for any $D\in \mathbb{Z}^n$ and $n$ not too small, it returns a range $\tilde{\mathcal{R}}(D)$ such that $|\tilde{\mathcal{R}}(D)|\leq 4\cdot \gamma(D)$, and $\left|D\cap\overline{\tilde{\mathcal{R}}(D)}\right| = O\left(\frac{1}{\varepsilon}\log\log\left(\gamma(D)\right)\right)$.
\end{theorem}

We can now invoke existing finite-domain empirical mean estimators \cite{nissim2007smooth,asi2020instance,huang2021instance} using $\tilde{\mathcal{R}}(D)$ as the domain, but it turns out that using $\tilde{\mathcal{R}}(D)$ directly with the \textit{clipped mean estimator} (Section \ref{sec:clipped_mean}) yields an even better result:

\begin{theorem}[Theorem~\ref{th:err_infinite_mean}, informal]
\label{th:err_infinite_mean_intro}
There exists an $\varepsilon$-DP mechanism such that for any $D\in\mathbb{Z}^n$ and $n$ not too small, 
it returns a $\tilde{\mu}(D)$ such that $\left|\tilde{\mu}(D)-\mu(D)\right| = O\left(\frac{\gamma(D)}{\varepsilon n}\log\log\left(\gamma(D)\right)\right)$.
\end{theorem}

\cite{huang2021instance,vadhan2017complexity} show that the width $\gamma(D)$ is an instance-specific lower bound.  More precisely, any mean estimator under DP (pure or not) has to incur an error of $\Omega(\gamma(D)/n)$ on $D$ or one of its \textit{in-neighbors} (see Section~\ref{sec:opt} for more details), so a result like Theorem \ref{th:err_infinite_mean_intro} can be considered instance-optimal, where the extra $O(\log\log(\gamma(D))/\varepsilon)$ factor is the \textit{optimality ratio}.  In contrast, the optimality ratio in \cite{huang2021instance} is $O(\log N/\varepsilon)$\footnote{The optimality ratio stated in \cite{huang2021instance} is $O(\sqrt{\log N /\rho})$, which holds under $\rho$-CDP; for pure DP, it is $O(\log N /\varepsilon)$}.  Thus, we obtain an exponential improvement even in the finite-domain case.  Furthermore, we show that the optimality ratio cannot be better than $O(\log\log N/\varepsilon)$ for all $D$ in the finite-domain case:

\begin{theorem}[Theorem~\ref{th:opt_mean}]
\label{th:opt_mean_intro}
For any $\varepsilon$, any integer $N\geq 1$, $n>\log\log_2 N/\varepsilon$, and any $\varepsilon$-DP mechanism $\mathcal{M}:[N]^n\rightarrow\mathbb{R}$, there exists $D\in[N]^n$, such that $|\mathcal{M}(D) - \mu(D)|\geq\frac{\gamma(D)}{3\varepsilon n}\log\log_2(N)$.
\end{theorem}

For quantile estimation, there exists a finite-domain estimator (Section \ref{sec:INV}) that achieves a rank error of $O(\log N/\varepsilon)$. Invoking it with $\tilde{\mathcal{R}}(D)$ immediately yields:

\begin{theorem}[Theorem~\ref{th:err_infinite_quantile}, informal]
\label{th:err_infinite_quantile_intro}
There exists an $\varepsilon$-DP mechanism such that for any $D\in \mathbb{Z}^n$, any $1\leq \tau\leq n$, and $n$ not too small, it returns a value $\Tilde{X}_{\tau}$ such that $X_{\tau-t} \leq \Tilde{X}_{\tau}\leq X_{\tau+t},\footnote{Define $X_{i}=X_{n}$ for $i>n$ and $X_{i}=X_{1}$ for $i<1$.}$
for some $t = O\left(\frac{1}{\varepsilon}\log\left(\gamma(D)\right)\right)$.
\end{theorem}

In the finite-domain case, it is known that the rank error has to be $\Omega(\log N / \varepsilon)$ for at least one $D$, by a reduction from the \textit{interior point problem} \cite{beimel2010bounds,bun2015differentially}.  In contrast, our error guarantee is a more instance-specific one, which is also worst-case optimal in the finite-domain case.

In addition, it is worth pointing out that sum estimation is equivalent to answering self-join-free aggregation queries in a relational database under user-level privacy protection~\cite{dong2022r2t}, which has been widely researched in database community. In that problem, the state-of-the-art algorithm~\cite{dong2022r2t} achieves the error $O(\frac{\mathrm{rad}(D)}{\varepsilon }\log(N)\log\log(N))$ and also requires a domain assumption $N$.  Consequently, our result also yields a significant in that problem.

\subsubsection{Statistical Mean Estimation}
\label{sec:our_results_mean}

Next, we move onto the statistical setting, where $D$ is an i.i.d.\ sample drawn from some arbitrary, unknown $\mathcal{P}$. Before we can apply our infinite-domain empirical mean estimator (Theorem \ref{th:err_infinite_mean_intro}), we have to discretize $\mathbb{R}$.  Since the sampling error is already $O(\sigma/\sqrt{n})$, a bucket size of $b\le \sigma/n$ would suffice.  However, $\sigma$ is not known; actually, estimating $\sigma$ is another mean estimation problem.  Under assumption A2, prior work \cite{karwa2018finite,BunS19,KamathLSU19,kamath2020differentially} simply used $b=\sigma_{\min} / n$ as the bucket size.  Without any assumptions, we seek to find a privatized lower bound of $\sigma$ and use that as the bucket size.  After that, we can apply Theorem \ref{th:err_infinite_mean_intro}, but this leads to sub-optimal errors in the statistical setting.  The reason is that in the empirical setting, we wish to minimize the number of points outside $\tilde{\mathcal{R}}(D)$, which translates into the optimality ratio.  When $D$ is an i.i.d.\ sample, the points in $D$ are more well-behaved and we can use a smaller $\tilde{\mathcal{R}}(D)$ to clip $D$ more aggressively. Our idea is thus to find $\tilde{\mathcal{R}}(D')$ on a sub-sample $D'$ of $D$ and apply the clipped mean estimator. It turns out $|D'| = \varepsilon n$ is the right sub-sample size, which yields our main result on a universal private mean estimator:

\begin{theorem}[Theorem~\ref{th:err_mean}, informal]
\label{th:err_mean_intro}
There exists an $\varepsilon$-DP mechanism such that for any $\mathcal{P}$, given $D\sim \mathcal{P}^n$, if 
\begin{equation}
\label{eq:n_LB_intro}
n>\Omega\left(\frac{1}{\varepsilon}\log\log{1\over\varphi(1/16)}+\frac{1}{\varepsilon}\log\log\left(\mathrm{IQR}\right)+\frac{1}{\varepsilon}\log{|\mu|+\sigma+\gamma(\varepsilon n) \over \varphi(1/16)}\right),
\end{equation}
then it returns a $\tilde{\mu}$ such that
\begin{equation}
|\mu-\tilde{\mu}|
= O\left(\min_{\xi\geq10\cdot\gamma\left(\varepsilon n\right)+2\sigma} \left(\left|\E\left[X<\mu-\xi\right]+\E\left[X>\mu+\xi\right]\right|+\frac{\xi}{\varepsilon n}\log\log{\gamma(\varepsilon n) \over \varphi(1/16)}\right) +{\sigma \over \sqrt{n}} \right). \label{eq:mean_error_intro}
\end{equation}
\end{theorem}

The formal definitions of $\varphi(1/16)$, $\gamma(\varepsilon n)$, $\E\left[X<\mu-\xi\right]$, and $\E\left[X>\mu+\xi\right]$ are given in Section~\ref{sec:notation}. Roughly speaking,  $\varphi(1/16)$ is the minimum width of any interval with a probability mass $1/16$, which is strictly positive for any continuous distribution $\mathcal{P}$. This term is required due to the searching for a proper bucket size. For all well-behaved $\mathcal{P}$, $\varphi(1/16) = \Theta(\sigma)$, but it may get arbitrarily small (e.g., when $f$ has a very narrow and high peak), which we call ill-behaved.  Nevertheless, we would like to stress that (1) our algorithm does not need to know $\varphi(1/16)$ \textit{a priori} (the analysis needs it \textit{a posteriori}); (2) our dependency on $1/\varphi(1/16)$ will be logarithmic or even $\log\log$; and (3) we did not try to optimize the constant $1/16$.  $\gamma(\varepsilon n)$ is a constant-probability bound on $\gamma(D') = X'_{\varepsilon n} - X'_1$ when $D'$ is a random sample of size $\varepsilon n$ drawn from $\mathcal{P}$, while $\E\left[X<\mu-\xi\right]$ and $\E\left[X>\mu+\xi\right]$ are the contributions to $\mu$ from the regions outside $[\mu-\xi, \mu+\xi]$, which correspond to (part of) the bias in $\tilde{\mu}$.  The last  term in the $\min\{\dots\}$ of \eqref{eq:mean_error_intro} is the DP noise (both bias and variance).  Importantly, the achieved error is the best bias-variance trade-off over all possible $\xi > 10\cdot \gamma(\varepsilon)+2\sigma$.  The last term in \eqref{eq:mean_error_intro} is the sampling error, which exists even in the non-private setting, so it does not depend on $\varepsilon$. 

Most prior works in the statistical setting state their results in terms of sample complexity, namely, what is the required sample size $n$ for achieving error $\alpha$.  Our lower bound requirement \eqref{eq:n_LB_intro} on $n$ easily translates into a term in the sample complexity, but it is cumbersome to rewrite \eqref{eq:mean_error_intro} due to the use of $\gamma(\varepsilon n)$ and the $\min_{\xi}$. 
To facilitate the comparison, below we relax $\gamma(\varepsilon n)$ appropriately and consider some fixed $\xi$.  This will result in simpler (but possibly looser) versions of Theorem \ref{th:err_mean_intro} in terms of the sample complexity.  We may also use the $\tilde{O}$ notation to suppress polylogarithmic factors in $n, {1 \over \alpha}, \log |\mu|, \log \sigma, \log{1 \over \varphi(1/16)}, \log R, \log{\sigma_{\max} \over \sigma_{\min}}$.

\paragraph{Gaussian distributions.}
If $\mathcal{P}$ is a Gaussian, then $\varphi(1/16) = \Theta(\sigma)$ and $\gamma(\varepsilon n) = \tilde{O}\left(\sigma \sqrt{\log(\varepsilon n)}\right)$.  We fix $\xi= c \cdot \sigma \sqrt{\log(\varepsilon n)}$ for some large constant $c$. Then Theorem \ref{th:err_mean_intro} simplifies into:

\begin{theorem}[Theorem~\ref{th:err_mean_gaussian}]
\label{th:err_mean_gaussian_intro}
For any Gaussian $\mathcal{P}$ and any $\alpha>0$, the $\varepsilon$-DP mechanism from Theorem \ref{th:err_mean_intro} takes $n=\tilde{O}\left(\frac{1}{\varepsilon}\log{|\mu|\over\sigma}+ \frac{\sigma^2}{\alpha^2}+\frac{\sigma}{\varepsilon \alpha}\right)$ samples and returns a $\tilde{\mu}$ such that $|\tilde{\mu}-\mu| \le \alpha$.
\end{theorem}

For Gaussian mean, \cite{karwa2018finite} and \cite{KamathLSU19,biswas2020coinpress} gave two $\varepsilon$-DP mechanisms under A1/A2. Their sample complexities are $n=\tilde{O}\left(\frac{1}{\varepsilon}\log {R \over \sigma_{\mathrm{min}}}+ \frac{\sigma^2}{\alpha^2}+\frac{\sigma}{\varepsilon \alpha}\right)$ and $n=\tilde{O}\left(\frac{1}{\varepsilon}\log {R \over \sigma}+ {1\over \varepsilon}\log \frac{\sigma_{\mathrm{max}}}{\sigma_{\mathrm{min}}} + \frac{\sigma^2}{\alpha^2}+\frac{\sigma}{\varepsilon \alpha}\right)$, respectively, both  inferior to Theorem \ref{th:err_mean_gaussian_intro}. 

\cite{karwa2018finite} show that $\Omega\left(\frac{1}{\varepsilon}\log {R \over \sigma}+ \frac{\sigma^2}{\alpha^2}+\frac{\sigma}{\varepsilon \alpha}\right)$ samples are necessary.
In fact, what they have proved is a worst-case lower bound, i.e., for any $R$, $\sigma$, and any $\varepsilon$-DP mechanism $\mathcal{M}$, there exists a Gaussian distribution $\mathcal{P}$ with $|\mu_{\mathcal{P}}|\le R, \sigma_{\mathcal{P}} = \sigma$ such that $\mathcal{M}$ requires  $\Omega\left(\frac{1}{\varepsilon}\log {R \over \sigma}+ \frac{\sigma^2}{\alpha^2}+\frac{\sigma}{\varepsilon \alpha}\right)$ samples to estimate $\mu_{\mathcal{P}}$ within an error of $\alpha$.  Our mechanism indeed requires this many samples on a $\mathcal{P}$ with $|\mu_{\mathcal{P}}| = R, \sigma_{\mathcal{P}} = \sigma$, thus no contradiction.  

\paragraph{Heavy-tailed distributions.}
Next, we consider the case where $\mathcal{P}$ has a finite $k$th central moment $\mu_k$ for some $k\ge 2$. In this case, we have $\gamma(\varepsilon n)<O\left((\varepsilon n \mu_k)^{1/k}\right)$.  Fixing $\xi= c\cdot (\varepsilon n \mu_k)^{1/k}$ for some large $c$, we can show that Theorem \ref{th:err_mean_intro} simplifies to:

\begin{theorem}[Theorem~\ref{th:err_mean_heavy}]
\label{th:err_mean_heavy_intro}
For any $\mathcal{P}$ with $k$-th central moment $\mu_k$ for some $k\ge 2$, and any $\alpha>0$, the $\varepsilon$-DP mechanism in Theorem \ref{th:err_mean_intro} takes 
\begin{equation}
\label{eq:err_mean_heavy_intro}
n=\tilde{O}\left( \frac{1}{\varepsilon}\log{|\mu|+(\varepsilon \mu_k)^{1/k} \over \varphi(1/16)}+ \frac{\sigma^2}{\alpha^2}+\frac{\mu_k^{1/(k-1)}}{\varepsilon \alpha^{k/(k-1)}}\right)
\end{equation}
samples and returns a $\tilde{\mu}$ such that $|\tilde{\mu}-\mu| \le \alpha$.
\end{theorem}

As our universal estimator does not need to know $k$ and $\mu_k$, Theorem \ref{th:err_mean_heavy_intro} actually holds for any $(k, \mu_k)$, and the bound should really be the infimum over all $k$.  In particular, if $\mathcal{P}$ is Gaussian, for which $\mu_k \le \sigma^k (k-1)!!$ for all $k$, Theorem \ref{th:err_mean_heavy_intro} essentially degenerates into Theorem \ref{th:err_mean_gaussian_intro} by setting $k$ to a large constant.  Anyhow, we would still state Theorem \ref{th:err_mean_heavy_intro} for a single $k$ for ease of comparison with prior work.  Also note that, as $k$ gets smaller, the privacy term $\frac{\mu_k^{1/(k-1)}}{\varepsilon \alpha^{k/(k-1)}}$ becomes more significant compared with the sampling error $\frac{\sigma^2}{\alpha^2}$.  This is intuitive: As $\mathcal{P}$ more spreads out, the individual values in the sample become more important, hence a higher cost for privacy.  For $k=2$, the privacy term would dominate the sampling error for all $\varepsilon \le 1$.

For heavy-tailed distributions, the previous $\varepsilon$-mechanism~\cite{kamath2020private} requires A1/A2 (for A2, their assumption is that $\mu_k \le \bar{\mu}_k \le R^k$ for given $k, \bar{\mu}_k$). Their sample complexity is 
\begin{equation}
\label{eq:heavy_tail_prior}
n = \tilde{O}\left(\frac{1}{\varepsilon}\log {R \over \bar{\mu}^{1/k}_k}+\frac{\sigma^2}{\alpha^2}+\frac{\bar{\mu}_k^{1/(k-1)}}{\varepsilon \alpha^{k/(k-1)}}\right)
\end{equation}

The sampling error term ${\sigma^2 \over \alpha^2}$ in \eqref{eq:heavy_tail_prior} is the same as the one in \eqref{eq:err_mean_heavy_intro}.  For the privacy term (the last term) in \eqref{eq:heavy_tail_prior} to match that in \eqref{eq:err_mean_heavy_intro}, they will need $\bar{\mu}_k$ to be a constant-factor approximation of $\mu_k$, which is not known how to obtain in a DP fashion.  In fact, if $\mu_{2k} = \infty$, there is no way to obtain such a $\bar{\mu}_k$ even in the non-private setting other than by assumption.  Assuming such a $\bar{\mu}_k = O(\mu_k)$ is given, it remains to compare $O\left(\log{|\mu|+(\varepsilon \mu_k)^{1/k} \over \varphi(1/16)}\right)$ and $O\left(\log {R \over \bar{\mu}^{1/k}_k}\right) = O\left(\log {R \over \mu^{1/k}_k}\right)$.  Since $|\mu|\le R, \mu_k^{1/k}\le R$, the former is always better unless $\mathcal{P}$ is ill-behaved: $\log{1\over \varphi(1/16)}=\omega\left( \log{1\over \mu_k^{1/k}}\right)$, i.e., $\varphi(1/16)$ is more than polynomially smaller than $\mu_k^{1/k}$. \cite{kamath2020private} also prove that $\Omega\left(\frac{\bar{\mu}_k^{1/(k-1)}}{\varepsilon \alpha^{k/(k-1)}}\right)$ samples are necessary.  Similar to the argument in the Gaussian case, this lower bound is worst-case.  It does not imply that this many samples are needed for every $\mathcal{P}$, or that the $\mu_k \le \bar{\mu}_k$ assumption is needed \textit{a priori}.

\paragraph{Arbitrary distributions.}
If $\mathcal{P}$ only has finite $\mu_2 = \sigma^2$, this corresponds to the most difficult distributions.  Note that in this case, the sample complexity of \cite{kamath2020private} becomes
\begin{equation}
\label{eq:heavy_tail_prior2}
n = \tilde{O}\left(\frac{1}{\varepsilon}\log {R \over \sigma_{\max}}+\frac{\sigma^2}{\alpha^2}+\frac{\sigma_{\max}^2}{\varepsilon \alpha^2}\right) =
\tilde{O}\left(\frac{1}{\varepsilon}\log {R \over \sigma_{\max}}+\frac{\sigma_{\max}^2}{\varepsilon \alpha^2}\right)
\end{equation}
For this problem, \cite{BunS19} proposed a different mean estimator under A1/A2 with the sample complexity
\begin{equation}
\label{eq:heavy_tail_prior3}
n=\tilde{O}\left(\frac{1}{\varepsilon}\log {R \over \sigma_{\min}}+\frac{\sigma^2}{\varepsilon^2\alpha^2}+\frac{\sigma^2}{\varepsilon\alpha^2}\log{R \over \sigma_{\min}}\right).\footnote{The result in \cite{BunS19} is claimed under CDP, which leads to a result under pure-DP by changing a distribution of noise.}
\end{equation}

These two results do not dominate each other.  
If the given $\sigma_{\max}$ is a constant-factor approximation of $\sigma$, then \eqref{eq:heavy_tail_prior2} is better than \eqref{eq:heavy_tail_prior3}; otherwise, \eqref{eq:heavy_tail_prior2} can be arbitrarily worse than \eqref{eq:heavy_tail_prior3}.  Note that again there is no way to obtain a good $\sigma_{\max}$ other than by assumption for a $\mathcal{P}$ with $\mu_4 = \infty$.

Meanwhile, our algorithm is better than both \cite{kamath2020private} and \cite{BunS19} except for ill-behaved $\mathcal{P}$.  Setting $k=2$, \eqref{eq:err_mean_heavy_intro} becomes
\begin{equation}
\label{eq:arbitrary_intro}
n=\tilde{O}\left( \frac{1}{\varepsilon}\log{|\mu|+\sqrt{\varepsilon}\sigma \over \varphi(1/16)}+\frac{\sigma^2}{\varepsilon \alpha^2}\right).
\end{equation}
We have already compared with \cite{kamath2020private} above for a general $k$. For the comparison with \cite{BunS19}, in addition to achieving pure DP, we see that the second term in \eqref{eq:arbitrary_intro} is strictly better than the last two terms in  \eqref{eq:heavy_tail_prior3}.  The first term in \eqref{eq:arbitrary_intro} is also better than that in \eqref{eq:heavy_tail_prior3} in most reasonable cases, unless $\mathcal{P}$ is ill-behaved ($\varphi(1/16) \ll \sigma$) or a very small $R$ is given (which would make the mean estimation problem meaningless).

\subsubsection{Variance Estimation}
\label{sec:our_results_variance}

For variance estimation, we first use the standard technique of randomly pairing up the elements in $D$.  For each pair $(X, X')$, compute $Z=(X-X')^2$, and let $H=\{Z_1,Z_2,\cdots,Z_{n/2}\}$ be the resulting $Z$'s. Since $\E[Z]=2\sigma^2$, the problem boils down to estimating $\E[Z]$. As our mean estimator is universal, we can apply it directly without worrying about the distribution of $Z$. In fact, the algorithm is even simpler, since the range of $Z$ is zero-centered thus easier to find. The following is our main result on universal variance estimation:

\begin{theorem}[Theorem~\ref{th:err_var}, informal]
\label{th:err_var_intro}
There exists an $\varepsilon$-DP mechanism such that for any $\mathcal{P}$, given $D\sim \mathcal{P}^n$, if $n>\Omega\left(\frac{1}{\varepsilon}\log\log{1 \over\varphi(1/16)}+\frac{1}{\varepsilon}\log\log\left(\mathrm{IQR}\right) \right)$,
then it returns a $\tilde{\sigma}^2$ such that
\begin{align*}
|\sigma^2-\tilde{\sigma}^2|=&O\left(\min_{\xi\geq 5\cdot \gamma(\varepsilon n)^2+2\sigma^2}\left(\left|\E\left[Z>2\sigma^2+\xi\right]\right|+\frac{\xi}{\varepsilon n} \log \log{\gamma(\varepsilon n) \over \varphi(1/16)}\right) + \sqrt{\frac{\mu_4}{n}}\right).
\end{align*}
\end{theorem}

Going through similar exercises, we obtain simplified results in terms of the sample complexity for specific distributions.

\paragraph{Gaussian distributions.} 
For Gaussian distributions, we have $\mu_4 = O(\sigma^4)$, and the simplified result is:
\begin{theorem}[Theorem~\ref{th:err_var_gaussian}]
\label{th:err_var_gaussian_intro}
For any Gaussian $\mathcal{P}$, and any $\alpha>0$, the $\varepsilon$-DP mechanism from Theorem \ref{th:err_var_intro} takes 
\begin{equation}
\label{eq:err_var_gaussian_intro}
n=\tilde{O}\left( {1\over \varepsilon} \max\left\{ \log\log\sigma, \log\log{1\over\sigma}\right\} + \frac{\sigma^4}{\alpha^2}+\frac{\sigma^2}{\varepsilon \alpha}\right)
\end{equation}
samples and returns a $\tilde{\sigma}^2$ such that $|\tilde{\sigma}^2-\sigma^2|\le\alpha$.
\end{theorem}

The last two terms are the same as for Gaussian mean estimation (Theorem \ref{th:err_mean_gaussian_intro}), except that $\sigma$ is replaced by $\sigma^2$.  The first term is more interesting, where we are able to reduce a $\log$ term to a $\log\log$.  This is exactly due to the simplification mentioned above: finding the width of the range enclosing $\E[Z]$ is exponentially easier than finding its location.  Meanwhile, since the error in $\tilde{\sigma}^2$ is relative to $\sigma^2$ itself (in contrast, the error in $\tilde{\mu}$ is relative to $\sigma$), we have to prepare for the case where $\sigma$ is very small, hence the $\log\log{1\over \sigma}$ term in \eqref{eq:err_var_gaussian_intro}.

There are two existing Gaussian variance estimators that do not dominate each other.  \cite{karwa2018finite} under A1/A2 achieve a sample complexity of
\begin{equation}
    \label{eq:err_var_gaussian_prior1}
    n = \tilde{O}\left({1\over \varepsilon} \log{R \over \sigma_{\min}} + \frac{1}{\varepsilon}\log\log{\sigma_{\max} \over \sigma_{\min}} + {\sigma^4 \over \alpha^2} + {\sigma^4 \over \varepsilon \alpha^2}\right),
\end{equation}
while \cite{KamathLSU19,biswas2020coinpress} under A2 achieve sample complexity
\begin{equation}
    \label{eq:err_var_gaussian_prior2}
n=\tilde{O}\left(\frac{1}{\varepsilon}\log {\sigma_{\max} \over \sigma_{\min}}+\frac{\sigma^4}{\alpha^2}+\frac{\sigma^2}{\varepsilon \alpha}\right).
\end{equation}

These two results are incomparable: \eqref{eq:err_var_gaussian_prior2} has an (almost) quadratically better privacy term (the last term) than \eqref{eq:err_var_gaussian_prior1}, but its dependency on ${\sigma_{\max} \over \sigma_{\min}}$ is exponentially worse.  On the other hand, \eqref{eq:err_var_gaussian_intro} is better than both, unless A2 already gives a tight range for $\sigma$.  In fact, if we are also given $\sigma_{\min}$, we can scale the data by ${1 \over \sigma_{\min}}$, and \eqref{eq:err_var_gaussian_intro} would further simplify to $n=\tilde{O}\left( {1\over \varepsilon} \log\log{\sigma \over \sigma_{\min}}+ \frac{\sigma^4}{\alpha^2}+\frac{\sigma^2}{\varepsilon \alpha}\right)$,
which is always better than both \eqref{eq:err_var_gaussian_prior1} and \eqref{eq:err_var_gaussian_prior2}.

\paragraph{Heavy-tailed distributions.}
Theorem \ref{th:err_var_intro} can be simplified into the following bound in terms of the central moments:
\begin{theorem}[Theorem~\ref{th:err_var_heavy}]
\label{th:err_var_heavy_intro}
For any $\mathcal{P}$, and any $\alpha>0$, the $\varepsilon$-DP mechanism in Theorem \ref{th:err_var_intro} takes $n=\tilde{O}\left(\frac{\mu_4}{\alpha^2}+ \inf_{k\ge 4}\frac{\mu_k^{2/(k-2)}}{\varepsilon \alpha^{k/(k-2)}}\right)$ samples and returns a $\tilde{\sigma}^2$ such that $|\tilde{\sigma}^2-\sigma^2|\le \alpha$. 
\end{theorem}

This is the first private variance estimator for heavy-tailed distributions.

\subsubsection{IQR Estimation}
\label{sec:our_results_scale}

Our IQR estimator is very simple: Discretize $\mathbb{R}$ using an appropriate bucket size return $\tilde{X}_{3n/4} - \tilde{X}_{n/4}$ using Theorem \ref{th:err_infinite_quantile_intro}.  We show that it achieves the following sample complexity:

\begin{theorem}[Theorem \ref{th:err_scale}]
\label{th:IQR_intro}
There exists an $\varepsilon$-DP mechanism such that for any $\mathcal{P}$ and any $\alpha>0$, it takes
\begin{equation}\label{eq:IQR_intro}
n=\tilde{O}\left(\frac{1}{\varepsilon}\log{|\mu|+\sigma + \gamma(n) \over \varphi(1/16)}+\frac{1}{\varepsilon\alpha\cdot \theta(\alpha/4)}\log {\gamma(n) \over \varphi(1/16)}+\frac{1}{(\alpha\cdot\theta(\alpha/4))^2}+\frac{\mathrm{IQR}}{\alpha}\right)
\end{equation}
samples and returns an $\widetilde{\mathrm{IQR}}$ such that $|\widetilde{\mathrm{IQR}} - \mathrm{IQR}| \le \alpha$. 
\end{theorem}

Here, $\theta(\alpha)$ is the average value of $f(x)$ in an interval of width $\alpha$ near $F^{-1}(1/4)$ and $F^{-1}(3/4)$ (formal definition given in Section~\ref{sec:scale}).
The previous IQR estimator \cite{dwork2009differential} only satisfies $(\varepsilon,\delta)$-DP. Their sample complexity is\footnote{\cite{dwork2009differential} defines $\theta(\cdot)$ as the minimum value of $f(x)$ in a small interval near $F^{-1}(1/4)$ and $F^{-1}(3/4)$, but their proof  still works even if it is defined as the average value, which makes the result stronger.}
\begin{equation}
\label{eq:IQR_prior}
n = \tilde{O}\left(\frac{1}{\left(\theta(2 n^{-1/3})\right)^6}+\frac{1}{\mathrm{IQR}^3}+\frac{1}{\alpha^3} + \exp\left(\frac{\mathrm{IQR}}{\varepsilon \alpha}\right)\right).
\end{equation}

To simplify the comparison between \eqref{eq:IQR_intro} and \eqref{eq:IQR_prior}, we consider a well-behaved $\mathcal{P}$ where $\theta(\alpha) = \Omega(1/\mathrm{IQR})$ (e.g., for  Gaussians, we have $\theta(\alpha)=\Theta(1/\mathrm{IQR})=\Theta(1/\sigma)$ for all $\alpha \le \mathrm{IQR}$) and ignore the logarithmic terms. Then \eqref{eq:IQR_intro} simplifies to $\tilde{O}\left({\mathrm{IQR} \over \varepsilon \alpha} + {\mathrm{IQR}^2 \over \alpha^2}\right)$ while \eqref{eq:IQR_prior} becomes  $\tilde{O}\left(\frac{1}{\mathrm{IQR}^3} + \mathrm{IQR}^6+\frac{1}{\alpha^3}\right.$ $\left.+\exp\left(\frac{\mathrm{IQR}}{\varepsilon \alpha}\right)\right)$.  
Note that their sampling error $\mathrm{IQR}^6 + {1\over \alpha^2} \ge (\mathrm{IQR}^6)^{1/3} \cdot ({1\over \alpha^3})^{2/3} = {\mathrm{IQR}^2 \over \alpha^2}$, while their privacy term $\exp\left(\frac{\mathrm{IQR}}{\varepsilon \alpha}\right)$ is exponentially worse than ours.
In particular, we get the right convergence rate $\alpha \propto {1/ (\varepsilon n)}$ for the privacy noise, which agrees with that for $\mu$ and $\sigma^2$.  On the other hand, their rate is $\alpha \propto 1/(\varepsilon \log n)$.

\subsection{Other Related Work}
\label{sec:related_work}

Many works on mean estimators extend to higher dimensions, such as \cite{KamathLSU19,cai2019cost,bun2019private,biswas2020coinpress,kamath2020private,huang2021instance,aden2021sample,kamath2021private,hopkins2021efficient,brown2021covariance,liu2021differential,ashtiani2021private,kothari2021private}. Using the idea of \cite{huang2021instance} but replacing Gaussian mechanism with Laplace mechanism, we can extend our pure-DP estimator to the multivariate case. However, it does not get the optimal privacy term $\tilde{O}(d/(\varepsilon n))$.
In fact, the problem is open even under A1/A2/A3 (assuming multivariate Gaussians for A3). \cite{kamath2020private} achieve the optimal $\tilde{O}(d/(\varepsilon n))$ but their algorithm runs in exponential time; the mechanism in \cite{hopkins2021efficient} runs in polynomial time but its privacy error is $\tilde{O}(\sqrt{d/(\varepsilon n)})$. Besides, \cite{liu2021robust} propose a solution for robust mean estimation under differential privacy.  The mean estimation problem has also been studied in the \textit{local model} of DP~\cite{duchi2013local,duchi2018minimax,gaboardi2019locally,duchi2019lower,joseph2019locally}, which is also an interesting direction to look at.

Covariance estimation in high dimensions has also received a lot of attention. \cite{KamathLSU19,bun2019private,biswas2020coinpress} consider multivariate Gaussian distributions and make similar boundedness assumptions like A1/A2. \cite{aden2021sample,kamath2021private,liu2021differential,ashtiani2021private,kothari2021private} do not need such assumptions but they relax the privacy notion to approximate DP. \cite{dong2022differentially,amin2019differentially} study the covariance for the data with bounded norms, which is even stronger than A1/A2. \cite{chaudhuri2013near,dwork2014analyze,sheffet2017differentially,upadhyay2018price} study private PCA or OLS, which can also be used to estimate covariance. However, they also assume that the data have bounded norms.

In the empirical setting, worst-case optimality does not make sense for functions whose global sensitivity is very large or $\infty$, which is the case for the empirical mean $\mu(D)$ where $D$ is drawn from an unbounded domain.  Instance-optimality is thus more suitable, but as pointed out by \cite{asi2020instance}, strict instance-optimality is not possible, who therefore propose a natural relaxation by considering a small neighborhood.  Nevertheless, for functions like $\mu(D)$, the neighborhood has to be restricted to avoid degeneration into worst-case optimality \cite{huang2021instance}, as we explain in Section \ref{sec:opt}.  Besides, as mentioned in Section~\ref{sec:our_results_empirical}, our empirical estimator can be used to answer self-join-free aggregation queries in a relational database. Answering aggregation queries has also been extensively studied in database community~\cite{mcsherry2009privacy,narayan2012djoin,Palamidessi2012DifferentialPF,proserpio2014calibrating,arapinis2016sensitivity,johnson2018towards,kotsogiannis2019privatesql,tao2020computing,dong21:residual,dong2021nearly,dong2022r2t}. For more details, please see~\cite{dong2022r2t}.

\subsection{Open Problems}
\label{sec:limitations}
The first open problem, obviously, is to extend our result to high dimensions. As mentioned in Section~\ref{sec:related_work}, the challenge here is to achieve the optimal dependency on $d$.
Another interesting direction is that, since the utility guarantees of our estimators depend on the parameters of $\mathcal{P}$ to be estimated, we cannot output confidence intervals. One possible solution is to derive privatized upper bounds of these parameters, but it may be challenging to make these upper bounds as tight as possible.

\subsection{Organization}
\label{sec:org}

The paper is organized as follows. In Section~\ref{sec:preliminary}, we define certain concepts introduced above more formally, as well as some building blocks for our algorithm. In Section~\ref{sec:empirically_setting}, we present our estimators in the empirical setting. In Section~\ref{sec:mean},~\ref{sec:var} and~\ref{sec:scale}, we describe our universal estimators for mean, variance, and IQR respectively in the statistical setting.

\section{Preliminaries}
\label{sec:preliminary}

\subsection{Notation}
\label{sec:notation}

Given a multiset $D = \{X_1, \dots, X_n\} \in \mathbb{R}^n$ (we reorder $D$ such that $X_1\le \cdots \le X_n$), we introduce the following notation: Its \textit{support} is $\mathrm{supp}(D)$, \textit{range} is $\mathcal{R}(D) = [X_1, X_n]$, \textit{width} is $\gamma(D) = X_{n}-X_{1}$, and \textit{radius} is $\mathrm{rad}(D) = \max_i |X_i|$. It is clear that $\mathcal{R}(D)\subseteq [-\mathrm{rad}(D),\mathrm{rad}(D)]$, hence $\gamma(D)\le 2 \cdot \mathrm{rad}(D)$, but $\mathrm{rad}(D)$ can be arbitrary larger than $\gamma(D)$. For any $\mathcal{S}\subseteq \mathbb{R}$, let $\left|D\cap \mathcal{S}\right| =\left|\{1\le i\le n \mid X_i\in D\cap \mathcal{S}\}\right|$.

Given a continuous probability distribution $\mathcal{P}$ over $\mathbb{R}$, in addition to $\mu$, $\sigma^2$,  $\mathrm{IQR}$ defined in Section \ref{sec:intro}, we also need the following quantities: For any $k\geq 2$, the $\textit{$k$th-central moment}$ is $\mu_k = \E_{X\sim P}[|X-\mu|^{k}]$.  In particular, $\mu_2 = \sigma^2$. For any $\beta\in(0,1)$, the width of the \textit{highest density region} at level $\beta$ is
\[\varphi(\beta) = \inf\left\{a_2-a_1\left|a_1,a_2\in \mathbb{R},a_2>a_1,\int_{a_1}^{a_2}f(x)\,\mathrm{d}x \geq \beta\right.\right\}.\]
We will only need $\varphi(\beta)$ for some constant $\beta$.  Note that $\varphi(1/2) \le \mathrm{IQR} \le 4\sigma$ (the first inequality is by definition and the second is by Chebyshev's inequality).  For most $\mathcal{P}$, the three quantities are close (e.g., for a Gaussian $\mathcal{P}$, the three are all within a constant factor from each other), although the gap can be arbitrarily large for an ill-behaved $\mathcal{P}$.

For any $m\in \mathbb{N}$ and $\beta\in(0,1)$, define the \textit{$(m,\beta)$-statistical width} of $\mathcal{P}$ as
\[\gamma(m,\beta) = \inf \left\{ \lambda \in \mathbb{R} \left| \Pr_{D\sim \mathcal{P}^m}\left[ \gamma(D) \ge \lambda\right] \leq \beta \right.\right\}.\]
Note that
\[\gamma\left(2,\frac{3}{4}\right)\leq \mathrm{IQR}\leq \gamma\left(\log_{\frac{4}{3}}(2/\beta),\beta\right).\]
The first inequality is because for $X\sim \mathcal{P}$, with probability $\frac{1}{2}$, $X\in[F^{-1}(1/4),F^{-1}(3/4)]$; the second inequality follows from the fact that $X\in [-\infty,F^{-1}(1/4)]$ and $X\in[F^{-1}(3/4),\infty]$ each happens with probability $\frac{1}{4}$, plus a union bound.

For $X\in\mathcal{P}$ and any $x\in\mathbb{R}$, define
\[\E[X\lessgtr x]:= \E_{X\sim \mathcal{P}}\left[(X-x)\mathbb{I}(X\lessgtr x)\right].\]

Finally, we introduce the following shorthand: For any $a,b\in \mathbb{R}$, let $[a\pm b] := [a-b,a+b]$. For interval $[l,r]$ and $b\in \mathbb{R}$, let $[l,r]\pm b := [l-b,r+b]$.  Define $[N] := \{0,1,\dots, N\}$.

\subsection{Differential Privacy}
\label{sec:dp}

The DP definition has already been introduced in Section~\ref{sec:intro}. The following two properties of DP are well-known:

\begin{lemma} [Post Processing \cite{dwork2006calibrating}]
\label{lm:post_processing_dp}
If $\mathcal{M}:\mathcal{X}^n\rightarrow \mathcal{Y}$ satisfies $\varepsilon$-DP and $\mathcal{M}':\mathcal{Y}\rightarrow \mathcal{Z}$ is any randomized mechanism, then $\mathcal{M}'(\mathcal{M}(D))$ satisfies $\varepsilon$-DP.
\end{lemma}

\begin{lemma} [Basic Composition \cite{dwork2006calibrating}]
\label{lm:basic_composition_dp}
If $\mathcal{M}_1:\mathcal{X}^n\rightarrow \mathcal{Y}$ satisfies $\varepsilon_1$-DP and $\mathcal{M}_2:\mathcal{X}^n\times \mathcal{Y}\rightarrow \mathcal{Z}$ satisfies $\varepsilon_2$-DP, then $\mathcal{M}_2(D,\mathcal{M}_1(D))$ satisfies $(\varepsilon_1 + \varepsilon_2)$-DP.
\end{lemma}

For any function $Q$, its local sensitivity at $D$ is
\[\mathrm{LS}_Q(D) = \sup_{D\sim D'}|Q(D)-Q(D')|\]
and the global sensitivity is
\[\mathrm{GS}_Q =\sup_{D}\mathrm{LS}_Q(D).\]

A basic pure DP mechanism is the Laplace mechanism:

\begin{lemma}[Laplace Mechanism]
\label{lm:lap}
The mechanism 
\[\mathcal{M}_Q(D)=Q(D)+\mathrm{Lap}(\mathrm{GS}_Q/\varepsilon)\]
preserves $\varepsilon$-DP, where $\mathrm{Lap}(\mathrm{GS}_Q/\varepsilon)$ is a random variable drawn from the Laplace distribution with scale $\mathrm{GS}_Q/\varepsilon$.
\end{lemma}

Below we omit the subscript $Q$ if the context is clear.

We also need the following result, which shows that privacy can be amplified by sampling.

\begin{theorem}[Sampling Amplification~\cite{balle2018privacy}]
\label{th:sampling_amplification}
Let $\eta\in(0,1)$.  Given an $\varepsilon$-DP mechanism $\mathcal{M}$, define $\mathcal{S}_{\eta}$ as the operation of sampling $\eta n$ samples from $D$ without replacement, then $\mathcal{M}(\mathcal{S}_{\eta}(D))$ preserves $\left(\log(1+\eta(e^{\varepsilon}-1))\right)$-DP.
\end{theorem}

Note that for small $\varepsilon$, $\log(1+\eta(e^{\varepsilon}-1))\approx \eta \varepsilon$.

\subsection{Optimality}
\label{sec:opt}

The high-probability error of using $\mathcal{M}(D)$ to approximate $Q(D)$ is defined as
\[\mathrm{Err}(\mathcal{M},D,\beta) = \inf\left\{\lambda\in\mathbb{R}\mid \Pr\left[\left|\mathcal{M}(D)-Q(D)\right|\leq \lambda\right]\geq 1-\beta \right\}.\]
We often take $\beta$ as a constant, say $\beta=1/3$; in this case we simply write $\mathrm{Err}(\mathcal{M},D)$.

The Laplace mechanism is worst-case optimal.  However, for any function $Q$ with $\mathrm{GS}=\infty$, such as the empirical mean $\mu(D)$ when $D$ is taken from an unbounded domain, this optimality notion is meaningless.  For such a $Q$, \textit{instance-optimality} is more appropriate and much stronger:

\begin{definition}[Instance-optimality]
Define the per-instance lower bound:
\[\mathcal{L}_{\text{ins}}(D)=\inf_{\mathcal{M}'}\mathrm{Err}(\mathcal{M}',D).\] Then a DP mechanism $\mathcal{M}$ is $c$-instance-optimal if
\[\mathrm{Err}(\mathcal{M},D)\leq c\cdot \mathcal{L}_{\text{ins}}(D)\]
for every $D$, where $c$ is the optimality ratio, which may depend on $D$.
\end{definition}

Unfortunately, $\mathcal{L}_{\text{ins}}(D)=0$ for every $D$ due to the trivial DP mechanism $\mathcal{M}'(\cdot) \equiv Q(D)$. Thus, instance-optimal DP mechanisms do not exist unless $Q$ is trivial (i.e., $Q(D)$ is the same for all $D$).  Thus, the following natural relaxation has been proposed:

\begin{definition}[Neighborhood-optimality~\cite{asi2020instance,dong2021nearly}]
Define the neighborhood lower bound:
\[\mathcal{L}_{\text{nbr}}(D)=\inf_{\mathcal{M}'}\sup_{D':D' \sim D }\mathrm{Err}(\mathcal{M}',D').\]
Then a DP mechanism $\mathcal{M}$ is $c$-neighborhood-optimal if
\[\mathrm{Err}(\mathcal{M},D)\leq c\cdot \mathcal{L}_{\text{nbr}}(D),\]
for every $D$.
\end{definition}

\cite{vadhan2017complexity} show that $\mathcal{L}_{\text{nbr}}(D)=\Theta(\mathrm{LS}(D))$ for every $D$. For the empirical mean $\mu(D)$, we have $\mathrm{LS}(D)=\infty$, since one can change an element in $D$ arbitrarily to obtain $D'$.  Thus this relaxation is ``too much''.  To fix the issue, the idea is to restrict the neighborhood:

\begin{definition}[Inward-neighborhood-optimality~\cite{huang2021instance}]
Define the inward-neighborhood lower bound:
\[\mathcal{L}_{\text{in-nbr}}(D)=\inf_{\mathcal{M}'}\max_{D':D\sim D',\mathrm{supp}(D')\subseteq \mathrm{supp}(D) }\mathrm{Err}(\mathcal{M}',D').\]
Then a DP mechanism $\mathcal{M}$ is $c$-inward-neighborhood-optimal if 
\[\mathrm{Err}(\mathcal{M},D)\leq c\cdot \mathcal{L}_{\text{in-nbr}}(D),\]
for every $D$. 
\end{definition}

Note that the restricted neighborhood is only concerned with the utility of $\mathcal{M}$, which still has to meet the standard privacy requirement over all $D\sim D'$.  

For any function $Q$, $\mathcal{L}_{\text{in-nbr}}(D)$ is always finite, as $D$ can only have a finite number of inward neighbors (thus $\sup_{D'}$ is replaced by $\max_{D'}$). In particular, for the empirical mean $\mu(D)$, we have $\mathcal{L}_{\text{in-nbr}}(D) = \Theta(\gamma(D)/n)$ \cite{huang2021instance}.

\subsection{The Sparse Vector Technique}
\label{sec:svt}

\begin{algorithm}[t]
\label{alg:svt}
\LinesNumbered 
\caption{$\mathrm{SVT}$.}
\KwIn{$T$, $\varepsilon$, $Q_1(D),Q_2(D),\dots$}
$\tilde{T} \gets T+\mathrm{Lap}(2/\varepsilon)$\;
    \For{$i\gets 1,2,\dots$}{
    $\tilde{Q}_i(D)\gets Q_i(D)+\mathrm{Lap}(4/\varepsilon)$\;
        \If{$\tilde{Q}_i(D)>\tilde{T}$}{
            Break\;
        }
    }
\Return $i$\;
\end{algorithm}

The Sparse Vector Technique ($\mathrm{SVT}$)~\cite{dwork2009complexity} has as input a (possibly infinite) sequence of queries, $Q_1,Q_2,\dots$, where each query has global sensitivity $1$, and a threshold $T$. It aims to find the first query whose answer is above $T$. The detailed algorithm is given in Algorithm \ref{alg:svt}.  The $\mathrm{SVT}$ has been shown to satisfy $\varepsilon$-DP and enjoy the following error guarantee, which says that it will not stop until it gets close to $T$.

\begin{lemma}[\cite{dwork2014algorithmic}]
\label{lm:err_SVT}
Suppose there exists a $k_1$ less than the length of the query sequence such that for all $i=1,\dots,k_1$, $Q_i(D)\leq T-\frac{8}{\varepsilon}\log(2k_1/\beta)$.
Then with probability at least $1-\beta$, $\mathrm{SVT}$ returns an $i\ge k_1+1$.
\end{lemma}

However, as will be clear later, we will actually need a complementary result that guarantees that SVT will stop in time.  The following lemma gives such a result.  More importantly, it also yields a utility guarantee on the returned query.  

\begin{lemma}
\label{lm:simple_obs_SVT}
If there exists a $k_2$ such that $Q_{k_2}(D)\geq T+\frac{6}{\varepsilon}\log(2/\beta)$, then with probability at least $1-\beta$, $\mathrm{SVT}$ returns an $i\leq k_2$ such that $Q_i(D)\ge T-{6 \over \varepsilon}\log(2k_2/\beta)$.
\end{lemma}

\begin{proof}
First, by the tail bound of the Laplace distribution, with probability at least $1-\frac{\beta}{2}$, 
\begin{equation}
\label{lm:simple_obs_SVT_1}
|\tilde{T}-T|< \frac{2}{\varepsilon}\log(2/\beta).
\end{equation}
And with probability at least $1-\frac{\beta}{4}$,
\begin{equation}
\label{lm:simple_obs_SVT_2}
\tilde{Q}_{k_2}(D)> Q_{k_2}(D)- \frac{4}{\varepsilon}\log(2/\beta).
\end{equation}
By a union bound over (\ref{lm:simple_obs_SVT_1}) and (\ref{lm:simple_obs_SVT_2}), together with the given condition $Q_{k_2}(D)\geq T+\frac{6}{\varepsilon}\log(2/\beta)$, we have that with probability at least $1-{3\over 4}\beta$, $\tilde{Q}_{k_2}(D)> \tilde{T}$, which implies $i\leq k_2$.

To show $Q_i(D)\geq T-\frac{6}{\varepsilon}\log(2k_2/\beta)$, we also require the following condition, which will be shown to hold with probability at least $1-{\beta\over 4}$.  Consider each $j=1,\dots, k_2$. We have
\[\Pr\left[\tilde{Q}_{j}(D)\geq Q_{j}(D) + \frac{4}{\varepsilon}\log{2k_2\over\beta}\right] =\Pr\left[\mathrm{Lap}\left({4\over \varepsilon}\right)\geq  \frac{4}{\varepsilon}\log{2k_2\over\beta}\right]\le \frac{\beta}{4k_2}.\]
By a union bound over all $j$, we have that, with probability at least $1-\frac{\beta}{4}$, $\tilde{Q}_{j}(D)<Q_{j}(D) + \frac{4}{\varepsilon}\log(2k_2/\beta)$ for all $j$. By further combining with (\ref{lm:simple_obs_SVT_1}), we have $Q_i(D)\geq T-\frac{6}{\varepsilon}\log(2k_2/\beta)$.
\end{proof}

\subsection{The Inverse Sensitivity Mechanism}
\label{sec:INV}

The \textit{inverse sensitivity mechanism} ($\mathrm{INV}$)~\cite{asi2020instance} answers a query $Q$ with a discrete output range $\mathcal{Y}$. Given $Q$ and $D$, it returns a $y\in \mathcal{Y}$ such that there exists $D'$ not too far from $D$ and $Q(D')=y$. Concretely, for any $D$ and any $y\in \mathcal{Y}$, define the path length:
\[\mathrm{len}(Q,D,y) = \min_{D'}\{d(D,D'):Q(D')=y\},\]
where $d(D,D')$ is the number of different elements between $D$ and $D'$. $\mathrm{INV}$ instantiates the exponential mechanism with $\mathrm{len}$ as the score function:
\[\Pr(\mathrm{INV}(Q,D) = y) = \frac{\exp\left(-\varepsilon\cdot\mathrm{len}(Q,D,y)/2\right)}{\sum_{y'\in \mathcal{Y}}\exp\left(-\varepsilon\cdot \mathrm{len}(Q,D,y')/2\right)}.\]

The utility of INV follows from that of the exponential mechanism:

\begin{lemma}[\cite{asi2020instance}]
\label{lm:err_INV}
For any $D$ and $\beta$, with probability at least $1-\beta$, $\mathrm{INV}$ returns a $y$ such that there exists a $D'$ with $d(D,D')\leq \frac{2}{\varepsilon}\log(|\mathcal{Y}|/\beta)$ and $Q(D') = y$.
\end{lemma}

\begin{algorithm}
\label{alg:finite_domain_quantile}
\LinesNumbered 
\caption{$\mathtt{FiniteDomainQuantile}$.}
\KwIn{$D$, $\tau$, $\mathcal{X}$, $\varepsilon$, $\beta$}

\uIf{$\tau\leq \frac{2}{\varepsilon}\log\left(|\mathcal{X}|/\beta\right)$}{
    $\tau'=\frac{2}{\varepsilon}\log(|\mathcal{X}|/\beta)$\;
}
\uElseIf{$\tau\geq n- \frac{2}{\varepsilon}\log(|\mathcal{X}|/\beta)$}{
    $\tau'=n- \frac{2}{\varepsilon}\log\left(|\mathcal{X}|/\beta\right)$\;
}
\Else{
    $\tau'=\tau$\;
}
Run $\mathrm{INV}$ to find the $\tau'$-quantile of $D$.
\end{algorithm}

INV can be used to find a privatized quantile $X_{\tau}$ of $D$, if $D$ are taken from a finite ordered domain $\mathcal{X}$, where $\mathrm{len}(Q,D,y)$ is simply the number of elements of $D$ that are between $X_{\tau}$ and $y$.  Since $\mathrm{len}(Q,D,y)$ only changes when $y$ passes some element in $D$, the exponential mechanism can be implemented in $O(n)$ time (given $D$ sorted) as opposed to $O(|\mathcal{Y}|)$.  Some care has to be taken if $\tau$ is too close to $1$ or $n$, in which case INV may return something arbitrarily bad.  The details are shown in Algorithm~\ref{alg:finite_domain_quantile}, which enjoys a rank error guarantee:

\begin{lemma}
\label{lm:err_known_quantile}
Given $\varepsilon,\beta$ and a finite ordered domain $\mathcal{X}$, for any $D\in \mathcal{X}^n$ and any $1\leq \tau\leq n$, if $n>\frac{4}{\varepsilon}\log(|\mathcal{X}|/\beta)$, then with probability at least $1-\beta$, $\mathtt{FiniteDomainQuantile}$ returns an $\Tilde{X}_{\tau}$ such that
\[X_{\tau-\frac{4}{\varepsilon}\log(|\mathcal{X}|/\beta)} \leq \Tilde{X}_{\tau}\leq X_{\tau+\frac{4}{\varepsilon}\log(|\mathcal{X}|/\beta)}.\]
\end{lemma}

\begin{proof}
Follows from Lemma~\ref{lm:err_INV} and the fact that $|\tau -  \tau'| \le \frac{2}{\varepsilon}\log(|\mathcal{X}|/\beta)$.
\end{proof}

\cite{asi2020instance} also propose a continuous version of SVT and \cite{smith2011privacy} uses a similar idea to estimate a quantile in a bounded real value domain. However, as the domain is infinite, those algorithms do not have any utility guarantee in the empirical setting.

\subsection{The Clipped Mean Estimator}
\label{sec:clipped_mean}

A standard idea for dealing with an unbounded domain is to clip all values into a bounded range $[l,r]$.  Define 
\[\mathrm{Clip}\left(X,[l,r]\right)=\begin{cases}
l, & \text{if }X<l;\\
X, & \text{if }l\leq X \leq r;\\
r, &\text{if } X>r.
\end{cases}\]
Let
\[\mathrm{Clip}(D,[l,r]) = \{\mathrm{Clip}\left(X_i,[l,r]\right) \mid X_i\in D \}.\]
Then the clipped mean estimator is
\[\mathrm{ClippedMean}(D,[l,r]) = \mu(\mathrm{Clip}(D, [l,r])).\]

It is obvious that $\mathrm{ClippedMean}(\cdot,[l,r])$ has global sensitivity $(r-l)/n$.  Thus, $\mathrm{ClippedMean}(D,[l,r]) + \mathrm{Lap}\left({r-l \over \varepsilon n}\right)$ satisfies $\varepsilon$-DP.

\subsection{Inequalities}
\label{sec:concentration_bounds}

We will need the following inequalities:

\begin{lemma}[Chernoff's inequality]
\label{lm:chernoff}
Given $k$ independent Bernoulli random variables $X_1, \cdots, X_k$ and $\bar{X}=\sum_{i=1}^k X_i$ and $\E[\bar{X}]= \mu$, then for any $0\leq\rho\leq 1$,
\[\Pr\left[\bar{X}\leq (1-\rho)\mu\right]\leq \exp\left(-\frac{\rho^2\mu}{2}\right),\]
and
\[\Pr\left[\bar{X}\geq (1+\rho)\mu\right]\leq \exp\left(-\frac{\rho^2\mu}{3}\right).\]
\end{lemma}

\begin{lemma}[Bernstein's inequality]
\label{lm:bernstein}
Let $X_1, \cdots, X_k$ be independent, zero-mean random variables such that $|X_i|\leq t$ for all $i$, and $\sigma^2 = \sum_{i=1}^k\E[X_i^2]$, then for any $\lambda$,
\[\Pr\left[\left|\sum_{i=1}^kX_i \right|\geq \lambda \right]\leq 2\exp\left(-\frac{\lambda^2/2}{\sigma^2+t\lambda/3}\right).\]
\end{lemma}

\begin{lemma}[H\"{o}lder's inequality]
\label{lm:holder}
Given  two random variables $X_1$, $X_2$ over $\mathbb{R}$, for any $k>1$, 
\[\E\left[|X_1 X_2|\right]\leq \left(\E\left[|X_1|^k\right]\right)^{1/k}\left(\E\left[|X_2|^{k/(k-1)}\right]\right)^{1-1/k}.\]
\end{lemma}

\section{Problems in the Empirical Setting}
\label{sec:empirically_setting}

In this section, we design $\varepsilon$-DP mechanisms for estimating $\mu(D)$ and $X_{\tau}$, where $D$ is taken from $\mathbb{Z}$.  We will first obtain $\tilde{\mathcal{R}}(D)$, a privatized $\mathcal{R}(D)$, and then invoke INV and the clipped mean estimator.  It turns out that the instance optimality ratio crucially depends on how well $\tilde{\mathcal{R}}(D)$ approximates $\mathcal{R}(D)$. Finally, we discuss the case when the domain is $\mathbb{R}$.

\subsection{Estimate Radius}
\label{sec:radius}

Before estimating $\mathcal{R}(D)$, we first estimate $\mathrm{rad}(D)$. We will show how to obtain a $\widetilde{\mathrm{rad}}(D)$ such that $\widetilde{\mathrm{rad}}(D)\leq 2\cdot \mathrm{rad}(D)$ while  $[-\widetilde{\mathrm{rad}}(D), \widetilde{\mathrm{rad}}(D)]$ covers all but $O\left(\log\log(\mathrm{rad}(D))\right)$ elements of $D$.

Let $\mathrm{Count}(D,x)=\left|D\cap [-x,x]\right|$.  It is easy to see that $\mathrm{Count}(\cdot,x)$ has the global sensitivity $1$ for any $x$, while $\mathrm{rad}(D)$ is exactly the smallest $x$ such that $\mathrm{Count}(D, x)\ge n$.  Thus, a natural idea is to feed the query sequence $\mathrm{Count}(D, x)$ for $x=0,1,2,4,8,\dots$ to SVT with a threshold of $T=n$.  However, doing so suffers from the ``late stop'' problem, i.e., SVT may stop at a $\widetilde{\mathrm{rad}}(D)$ that is too large due to the exponential growth rate of $x$. On the other hand, reducing the growth rate increases the length of the query sequence, degrading the utility of SVT.  Inspired by Lemma~\ref{lm:simple_obs_SVT}, we use  $T=n-6\log(2/\beta)/\varepsilon$ so that $\mathrm{SVT}$ will stop at the ``right'' place. The details are shown in Algorithm~\ref{alg:radius}.

\begin{algorithm}
\LinesNumbered 
\label{alg:radius}
\caption{$\mathtt{InfiniteDomainRadius}$.}
\KwIn{$D$, $\varepsilon$, $\beta$}
$\tilde{i}= \mathrm{SVT}\left(n-\frac{6}{\varepsilon}\log(2/\beta),\varepsilon, \mathrm{Count}(D,0), \mathrm{Count}(D,2^0), \mathrm{Count}(D,2^1),\dots\right)$\;
\eIf{$\tilde{i}=1$}{
    $\widetilde{\mathrm{rad}}(D)=0$\;
}{
    $\widetilde{\mathrm{rad}}(D)=2^{\tilde{i}-2}$\;
}
\Return $\widetilde{\mathrm{rad}}(D)$\;
\end{algorithm}

The privacy of $\mathtt{InfiniteDomainRadius}$ follows from that of the SVT and the post-processing property of DP.  We analyze its utility below:

\begin{theorem}
\label{th:err_infinite_radius}
For any $D\in \mathbb{Z}^n$, with probability at least $1-\beta$, $\mathtt{InfiniteDomainRadius}$ returns a $\widetilde{\mathrm{rad}}(D)$ such that $\widetilde{\mathrm{rad}}(D)\leq 2\cdot \mathrm{rad}(D)$ 
and
\[\left|D\cap\overline{\left[-\widetilde{\mathrm{rad}}(D),\widetilde{\mathrm{rad}}(D)\right]} \right| = O\left(\frac{1}{\varepsilon}\log\left(\log\left(\mathrm{rad}(D)\right)/\beta\right)\right).\]
\end{theorem}

\begin{proof}
We consider two cases: $\mathrm{rad}(D) = 0$ and $\mathrm{rad}(D)\in[2^{j-1},2^{j}]$ for some $j\in \mathbb{N}$. In the first case, $\mathrm{Count}(D,0) = n $. By Lemma~\ref{lm:simple_obs_SVT}, with probability at least $1-\beta$, we have $\widetilde{\mathrm{rad}}(D) = 0$, and both conclusions hold.

In the second case, $\mathrm{Count}(D,2^j) = n$. Plugging in Lemma~\ref{lm:simple_obs_SVT} with $T= n-\frac{6}{\varepsilon}\log(2/\beta)$ and $k_2 = \log_2(2^j)+2$, we have, with probability at least $1-\beta$,
\begin{equation*}
\label{eq:th:err_range_1}
\widetilde{\mathrm{rad}}(D)\leq 2^j\leq 2 \cdot \mathrm{rad}(D).
\end{equation*}
and 
\[\mathrm{Count}(D,\widetilde{\mathrm{rad}}(D))\geq n-\frac{6}{\varepsilon}\log(2/\beta)-\frac{6}{\varepsilon} \log\left(2\left(\log_2(2^j)+2\right)/\beta\right).\]
\end{proof} 

\subsection{Estimate Range}
\label{sec:data_range}

To find a good privatized range $\tilde{\mathcal{R}}(D)$, we first search for an $\tilde{X}$ that is very likely located inside $\mathcal{R}(D)$, which can be done using INV to find a privatized median over a finite domain, as most data have been covered in $[-\widetilde{\mathrm{rad}}(D),\widetilde{\mathrm{rad}}(D)]$.  Next, we shift $D$ to be centered around $\tilde{X}$, and run $\mathtt{InfiniteDomainRadius}$ again. The detailed algorithm is shown in Algorithm~\ref{alg:range}.

\begin{algorithm}
\LinesNumbered 
\label{alg:range}
\caption{$\mathtt{InfiniteDomainRange}$.}
\KwIn{$D$, $\varepsilon$, $\beta$}
$\widetilde{\mathrm{rad}}(D) = \mathtt{InfiniteDomainRadius}(D,\frac{\varepsilon}{8},\frac{\beta}{3})$\;
$D' = \mathrm{Clip}\left(D,[-\widetilde{\mathrm{rad}}(D),\widetilde{\mathrm{rad}}(D)]\right)$\;
$\tilde{X} = \mathtt{FiniteDomainQuantile}\left(D',\frac{n}{2},\mathbb{Z}\cap \left[-\widetilde{\mathrm{rad}}(D),\widetilde{\mathrm{rad}}(D)\right],\frac{\varepsilon}{8}, \frac{\beta}{3}\right)$\;
$D'' = D - \tilde{X}$\;
$\widetilde{\mathrm{rad}}(D'') = \mathtt{InfiniteDomainRadius}(D'',\frac{3\varepsilon}{4},\frac{\beta}{3})$\;
$\tilde{\mathcal{R}}(D) = [\tilde{X}-\widetilde{\mathrm{rad}}(D''),\tilde{X}+\widetilde{\mathrm{rad}}(D'')]$\;
\Return $\tilde{\mathcal{R}}(D)$\;
\end{algorithm}

The privacy of $\mathtt{InfiniteDomainRange}$ follows from basic composition. Its utility is summarized by the following theorem:

\begin{theorem}
\label{th:err_infinite_range}
Given $\varepsilon$, $\beta$, for any $D\in \mathbb{Z}^n$, if
\[n>\frac{c_1}{\varepsilon}\log\left(\mathrm{rad}(D)/\beta\right),\]
where $c_1$ is a universal constant, then with probability at least $1-\beta$, $\mathtt{InfiniteDomainRange}$ returns a range $\tilde{\mathcal{R}}(D)$ such that
\[|\tilde{\mathcal{R}}(D)|\leq 4\cdot \gamma(D),\]
and
\[\left|D\cap\overline{\tilde{\mathcal{R}}(D)}\right| = O\left(\frac{1}{\varepsilon}\log\left(\log\left(\gamma(D)\right)/\beta\right)\right).\]
\end{theorem}

\begin{proof}
Let $D'' = \{X''_1,\dots,X'_n\}$. First, by Theorem~\ref{th:err_infinite_radius}, with probablity at least $1-\frac{\beta}{3}$,
\begin{equation}
\label{eq:th:err_infinite_radius_1}
\widetilde{\mathrm{rad}}(D)\leq 2\cdot \mathrm{rad}(D),
\end{equation}
and
\begin{equation}
\label{eq:th:err_infinite_radius_2}
\left|D\cap \overline{\left[-\widetilde{\mathrm{rad}}(D),\widetilde{\mathrm{rad}}(D)\right]}\right|=O\left( \frac{1}{\varepsilon}\log\left(\log\left(\mathrm{rad}(D)\right)/\beta\right)\right).
\end{equation}

Then, by Lemma~\ref{lm:err_known_quantile}, (\ref{eq:th:err_infinite_radius_1}), (\ref{eq:th:err_infinite_radius_2}), and setting $c_1$ sufficiently large, with probability at least $1-\frac{\beta}{3}$,
\begin{equation*}
X_1\leq \tilde{X}\leq X_n,
\end{equation*}
which further implies that 
\begin{equation}
\label{eq:th:err_infinite_radius_3}
\mathrm{rad}(D'')\leq \gamma(D).
\end{equation}

Finally, based on Theorem~\ref{th:err_infinite_radius}, with probability at least $1-\frac{\beta}{3}$,
\begin{equation}
\label{eq:th:err_infinite_radius_4}
\widetilde{\mathrm{rad}}(D'')\leq 2\cdot \mathrm{rad}(D''),
\end{equation}
and
\begin{equation}
\label{eq:th:err_infinite_radius_5}
\left|D''\cap\overline{\left[-\widetilde{\mathrm{rad}}(D''),\widetilde{\mathrm{rad}}(D'')\right]}
\right|\leq O\left(\frac{1}{\varepsilon}\log\left(\log\left(\mathrm{rad}(D'')\right)/\beta\right)\right).
\end{equation}

The first part of the conclusion follows from (\ref{eq:th:err_infinite_radius_3}) and (\ref{eq:th:err_infinite_radius_4}); the second part follows from (\ref{eq:th:err_infinite_radius_3}) and (\ref{eq:th:err_infinite_radius_5}).
\end{proof}

\subsection{Mean Estimation}
\label{sec:mean_empirical}

\begin{algorithm}[t]
\label{alg:infinite_domain_mean}
\LinesNumbered 
\caption{$\mathtt{InfiniteDomainMean}$.}
\KwIn{$D$, $\varepsilon$, $\beta$}
$\tilde{\mathcal{R}}(D) = \mathtt{InfiniteDomainRange}(D,\frac{4\varepsilon}{5},\frac{\beta}{2})$\;
$\tilde{\mu}(D) =\mathrm{ClippedMean}(D,\tilde{\mathcal{R}}(D)) +\mathrm{Lap}\left(5|\tilde{\mathcal{R}}(D)|/(\varepsilon n)\right)$\;
\Return $\tilde{\mu}(D)$\;
\end{algorithm}

With a good $\tilde{\mathcal{R}}(D)$, we can now do mean estimation over an infinite domain. The algorithm is shown in Algorithm~\ref{alg:infinite_domain_mean}.  Its privacy follows from basic composition, while its utility guarantee is as follows:

\begin{theorem}
\label{th:err_infinite_mean}
Given $\varepsilon$, $\beta$, for any $D\in\mathbb{Z}^n$, if
\[n>\frac{c_1}{\varepsilon}\log\left(\mathrm{rad}(D)/\beta\right),\]
where $c_1$ is a universal constant, then with probability at least $1-\beta$, $\mathtt{InfiniteDomainMean}$ returns a $\tilde{\mu}(D)$ such that
\[\left|\tilde{\mu}(D)-\mu(D)\right| = O\left(\frac{\gamma(D)}{\varepsilon n}\log\left(\log\left(\gamma(D)\right)/\beta\right)\right).\]
\end{theorem}

\begin{proof}
There are two sources of errors: the bias caused by the clipping and the Laplace noise.

By Theorem~\ref{th:err_infinite_range} and setting $c_1$ sufficiently large, we have that with probability at least $1-\frac{\beta}{2}$,
\begin{equation}
\label{eq:th:err_infinite_mean_1}
|\tilde{\mathcal{R}}(D)|\leq 4\cdot \gamma(D),
\end{equation}
and there are at most $O\left(\log\left(\log\left(\gamma(D)\right)/\beta\right)/\varepsilon\right)$ elements in $D$ outside $\tilde{\mathcal{R}}(D)$ and at least one inside $\tilde{\mathcal{R}}(D)$.  Thus, clipping each outlier causes at most $\gamma(D)/n$ bias and the total bias is $O\left(\log\left(\log\left(\gamma(D)\right)/\beta\right)\gamma(D)/(\varepsilon n)\right)$.

The Laplace noise can be bounded by plugging (\ref{eq:th:err_infinite_mean_1}) into the tail bound of the Laplace distribution, which yields $O\left(\log\left(1/\beta\right)\gamma(D)/(\varepsilon n)\right)$. 
\end{proof}

Recall from Section~\ref{sec:dp} and~\ref{sec:opt} that for the empirical mean $\mu(D)$, $\mathcal{L}_{\text{in-nbr}}(D)=\Omega(\gamma(D)/n)$  for every $D$. This means that $\mathtt{InfiniteDomainMean}$ is inward-neighborhood optimal with an optimality ratio of $c=O(\log\log(\gamma(D))/\varepsilon)$ for constant $\beta$. Below, we show that this $c$ is worst-case optimal in the finite-domain case.  In particular, it implies that the optimality ratio cannot be independent of $D$.

\begin{theorem}
\label{th:opt_mean}
For the empirical mean $\mu(D)$, given any $\varepsilon$, any integer $N\geq 1$, and any $n>\log\log_2(N)/\varepsilon$, for any $\varepsilon$-DP mechanism $\mathcal{M}:[N]^n\rightarrow\mathbb{R}$, there exists $D\in[N]^n$, such that
\[\mathrm{Err}(\mathcal{M},D)\geq\frac{\gamma(D)}{3\varepsilon n}\log\log_2(N).\]
\end{theorem}

\begin{proof}
We use a packing argument by constructing a sequence of $\log_2(N)+1$ datasets: $D^{(0)}$,$D^{(1)}$,$\dots$, $D^{(\log_{2}(N))}$. $D^{(0)}$ contains all $0$'s. For each $i=1,\dots,\log_2(N)$, $D^{(i)}$ is constructed by changing $\log\log_2(N)/\varepsilon$ number of $0$'s in $D^{(0)}$ to $2^i$. It can be verified that  
\begin{equation}
\label{eq:th:opt_mean_1}
\mu(D^{(i)}) = \frac{2^i}{\varepsilon n}\log\log_2(N).
\end{equation}

We argue that, for any $\varepsilon$-DP mechanism $\mathcal{M}$, there exists at least one $D^{(i)}$ such that
\begin{equation}
\label{eq:th:opt_mean_2}
\mathrm{Err}(\mathcal{M},D^{(i)})\geq \frac{\gamma(D^{(i)})}{3\varepsilon n}\log\log_2(N) =\frac{2^i}{3\varepsilon n}\log\log_2(N) .
\end{equation}
We prove this by contradiction. If (\ref{eq:th:opt_mean_2}) does not hold, then
\begin{align}
\nonumber
\frac{1}{3} \geq & \mathsf{Pr}[\mathcal{M}(D^{(0)})\neq 0]
\\
\geq &\sum_{1\leq i\leq \log_2(N)} \mathsf{Pr}\left[\mathcal{M}(D^{(0)})\in \left(\frac{2}{3}\cdot 2^i\cdot \frac{1}{\varepsilon n} \log\log_2(N),\frac{4}{3} \cdot 2^i \cdot \frac{1}{\varepsilon n} \log\log_2(N)\right)\right]
\label{eq:th:opt_mean_3}
\\
\nonumber
\geq & \sum_{1\leq i\leq \log_2(N)}\left( e^{-\varepsilon \log(\log_2(N))/\varepsilon} \mathsf{Pr}\left[\mathcal{M}(D^{(i)})\in \left(\frac{2}{3}\cdot 2^i\cdot \frac{1}{\varepsilon n} \log\log_2(N),\frac{4}{3} \cdot 2^i \cdot \frac{1}{\varepsilon n} \log\log_2(N)\right)\right]\right)
\\
\label{eq:th:opt_mean_4}
\geq &\log_2(N) \cdot \frac{1}{\log_2(N)}\cdot \frac{2}{3}
\\
\nonumber
=& \frac{2}{3},
\end{align}
which causes a contradiction. (\ref{eq:th:opt_mean_3}) is because the intervals $\left(\frac{2}{3}\cdot 2^i\cdot \frac{1}{\varepsilon n} \log\log_2(N),\frac{4}{3} \cdot 2^i \cdot \frac{1}{\varepsilon n} \log\log_2(N)\right)$ for all $i=1,\dots,\log_2(N)$ are disjoint. (\ref{eq:th:opt_mean_4}) follows from (\ref{eq:th:opt_mean_1}) and the hypothesis.
\end{proof}

\subsection{Quantile Estimation}
\label{sec:quantile}

\begin{algorithm}
\label{alg:infinite_domain_quantile}
\LinesNumbered 
\caption{$\mathtt{InfiniteDomainQuantile}$.}
\KwIn{$D$, $\tau$, $\varepsilon$, $\beta$}
$\tilde{\mathcal{R}}(D) = \mathtt{InfiniteDomainRange}(D,\frac{4\varepsilon}{5},\frac{\beta}{2})$\;
$D' = \mathrm{Clip}\left(D,\tilde{\mathcal{R}}(D)\right)$\;
$\tilde{X}_{\tau} = \mathtt{FiniteDomainQuantile}\left(D',\tau,\tilde{\mathcal{R}}(D),\frac{\varepsilon}{5},\frac{\beta}{2}\right)$\;
\Return $\tilde{X}_{\tau}$\;
\end{algorithm}

Similarly, to find a privatized quantile over an infinite domain, we invoke $\mathtt{FiniteDomainQuantile}$ with $\tilde{\mathcal{R}}(D)$. The algorithm is shown in Algorithm~\ref{alg:infinite_domain_quantile}. Its privacy is straightforward, while achieving $O(\log(\mathrm{rad}(D))/\varepsilon)$ rank error:

\begin{theorem}
\label{th:err_infinite_quantile}
Given $\varepsilon$, $\beta$, for any $D\in \mathbb{Z}^n$ and any $1\leq \tau\leq n$, if
\[n>\frac{c_1}{\varepsilon}\log\left(\mathrm{rad}(D)/\beta\right),\]
where $c_1$ is a universal constant, then with probability at least $1-\beta$, $\mathtt{InfiniteDomainQuantile}$ returns a value $\Tilde{X}_{\tau}$ such that
\[X_{\tau-t} \leq \Tilde{X}_{\tau}\leq X_{\tau+t},\]
where 
\[t = O\left(\frac{1}{\varepsilon}\log\left(\gamma(D)/\beta\right)\right).\]
\end{theorem}

\begin{proof}
By Theorem~\ref{th:err_infinite_range} and setting $c_1$ sufficiently large, with probability at least $1-\frac{\beta}{2}$, $|\tilde{\mathcal{R}}(D)|\leq 4\cdot \gamma(D)$ and $O\left(\log\left(\log\left(\gamma(D)\right)/\beta\right)/\varepsilon\right)$ values are clipped.

Under the condition of $|\tilde{\mathcal{R}}(D)|\leq 4\cdot \gamma(D)$, by Lemma~\ref{lm:err_known_quantile} and setting $c_1$ in the condition of $n$ properly, with probability at least $1-\frac{\beta}{2}$, $\mathtt{FiniteDomainQuantile}$ will only cause $ O\left(\log\left(\gamma(D)/\beta\right)/\varepsilon\right)$ rank error. The clipping does not increase this error asymptotically.
\end{proof}

The rank error of $\mathtt{FiniteDomainQuantile}$ is instance-specific, and worst-case optimal in the finite-domain case, by a reduction from the \textit{interior-point problem}. Here, given a dataset $D\in [N]^n$, we want to return any integer inside $\mathcal{R}(D)$. It has been shown that any $\varepsilon$-DP mechanism for the interior point problem requires $n=\Omega(\log(N)/\varepsilon)$ \cite{beimel2010bounds,bun2015differentially}.  Given a (finite-domain) quantile mechanism with rank error $t$, we would be able to solve the interior-point problem on datasets with $2t$ elements by returning the median.  Thus $\Omega(\log(N)/\varepsilon)$ is also a lower bound on the rank error.

\subsection{Extension to the Real Domain}
\label{sec:real_number_domain}

If $D$ are drawn from $\mathbb{R}$, we can invoke the algorithms above after discretizing $\mathbb{R}$ with a bucket size $\bucket$. This will introduce an additive error of $b$ to each value estimate and an extra $1/b$ factor to each count/rank estimate. This effect is slightly different for each particular problem, as summarized as follows.  We omit the rather straightforward proofs.

\begin{theorem}
\label{th:err_unbounded_radius}
Given $\varepsilon,\beta,\bucket$, for any $D\in \mathbb{R}^n$, with probability at least $1-\beta$, $\mathtt{InfiniteDomainRadius}$ returns a $\widetilde{\mathrm{rad}}(D)$ such that
\[\widetilde{\mathrm{rad}}(D)\leq 2\cdot \mathrm{rad}(D) + 3\bucket,\]
and
\[\left|D\cap\overline{\left[-\widetilde{\mathrm{rad}}(D),\widetilde{\mathrm{rad}}(D)\right]} \right| = O\left(\frac{1}{\varepsilon}\log\left(\log\left(\mathrm{rad}(D)/\bucket\right)/\beta\right)\right).\]
\end{theorem}

\begin{theorem}
\label{th:err_unbounded_range}
Given $\varepsilon,\beta,\bucket$, for any $D\in \mathbb{R}^n$, 
if
\[n>\frac{c_1}{\varepsilon}\log\left(\mathrm{rad}(D)/(\bucket \beta)\right),\]
where $c_1$ is a universal constant, then with probability at least $1-\beta$, $\mathtt{InfiniteDomainRange}$ returns a range $\tilde{\mathcal{R}}(D)$ such that
\[|\tilde{\mathcal{R}}(D)|\leq 4\cdot \gamma(D)+6\bucket ,\]
and
\[\left|D\cap\overline{\tilde{\mathcal{R}}(D)}\right| = O\left(\frac{1}{\varepsilon}\log\left(\log\left(\gamma(D)/\bucket \right)/\beta\right)\right).\]
\end{theorem}

\begin{theorem}
\label{th:err_unbounded_mean}
Given $\varepsilon,\beta,\bucket$, for any $D$, if
\[n>\frac{c_1}{\varepsilon}\log\left(\mathrm{rad}(D)/(b\beta)\right),\]
where $c_1$ is a universal constant, then with probability at least $1-\beta$, $\mathtt{InfiniteDomainMean}$ returns a $\tilde{\mu}(D)$ such that
\[\left|\tilde{\mu}(D)-\mu(D)\right| = O\left(\frac{\gamma(D)+b}{\varepsilon n}\log\left(\log\left(\gamma(D)/\bucket\right)/\beta\right)\right).\]
\end{theorem}

\begin{theorem}
\label{th:err_unbounded_quantile}
Given $\varepsilon,\beta,\bucket$, for any $D\in \mathbb{R}^n$, if
\[n>\frac{c_1}{\varepsilon}\log\left(\mathrm{rad}(D)/(b\beta)\right),\]
where $c_1$ is a universal constant, then with probability at least $1-\beta$, $\mathtt{InfiniteDomainQuantile}$ returns an $\Tilde{X}_{\tau}$ such that
\[X_{\tau-t}-\bucket \leq \Tilde{X}_{\tau}\leq X_{\tau+t}+\bucket ,\]
where 
\[t = O\left(\frac{1}{\varepsilon}\log\left(\gamma(D)/(\bucket \beta)\right)\right).\]
\end{theorem}

\section{Statistical Mean Estimation}
\label{sec:mean}

In this section, we consider the statistical mean estimation problem, i.e., given an i.i.d.\ sample $D \sim \mathcal{P}^n$ for an arbitrary, unknown $\mathcal{P}$ over $\mathbb{R}$, we wish to estimate $\mu_{\mathcal{P}}$.  The idea is conceptually simple: We first discrete  $\mathbb{R}$ with an appropriate bucket size $b$; then we invoke the empirical mean estimator over $\mathbb{Z}$.  For the first step, we find a lower bound on the IQR, denoted $\overline{\mathrm{IQR}}$, as the bucket size. For the second step, it turns out that directly invoking the empirical mean estimator in Theorem \ref{th:err_infinite_mean} results in sub-optimal errors in the statistical setting; instead, we shall use a tighter range to do the clipping. 

\subsection{Estimate a Lower Bound for $\mathrm{IQR}$}
\label{sec:lower_bound_IQR}

Prior work under A2 simply uses $b=\sigma_{\min}$ as the bucket size, which would be dominated by the sampling error.  In the absence of $\sigma_{\min}$, we seek to obtain a privatized lower bound of IQR, since $\mathrm{IQR}\le 4\sigma$.  Furthermore, recall   $\gamma\left(2,\frac{3}{4}\right)\leq \mathrm{IQR}$ (Section~\ref{sec:notation}), thus if we randomly draw two values $X$, $X'$ from $\mathcal{P}$, then with probability at least $\frac{1}{4}$, we have 
\[|X-X'|\leq \mathrm{IQR}.\]

Meanwhile, we do not want a bucket size too small.  We thus relate $|X-X'|$ with $\varphi(\cdot)$.

\begin{lemma}
\label{lm:lower_bound_diff}
For any $X,X'\in \mathcal{P}$, with probability at least $1-\frac{1}{8}$, we have
\[\varphi\left(\frac{1}{16}\right)\leq |X-X'|.\]
\end{lemma}

\begin{proof}
\begin{align}
\Pr\left(|X-X'|\leq \varphi\left(\frac{1}{16}\right)\right) &= \int^{-\infty}_{\infty}f(x)\int^{-\infty}_{\infty}f(x') \mathbb{I}\left(|x-x'|\leq \varphi\left(\frac{1}{16}\right)\right) \mathrm{d}x'\mathrm{d}x 
\nonumber
\\
& = \int^{-\infty}_{\infty}f(x)\int^{x-\varphi\left(\frac{1}{16}\right)}_{x+\varphi\left(\frac{1}{16}\right)} f(x') \,\mathrm{d}x'\mathrm{d}x 
\nonumber
\\
&\leq \int^{-\infty}_{\infty}f(x) \left( \frac{1}{16} \cdot 2 \right)\mathrm{d}x
\nonumber
\\
\nonumber
&= \frac{1}{8},
\end{align}
where the inequality is by the definition of $\varphi(1/16)$: Any interval with length $\varphi\left(\frac{1}{16}\right)$ can at most contain a probability mass of $\frac{1}{16}$.
\end{proof}

To amplify the success probability, we randomly group the elements in $D$ into pairs $(X,X')$ and let $G = \{Y_1,Y_2,\dots,Y_{n'}\}$ where $n'=n/2$ and $Y_i = |X-X'|$ for each pair.  Again, suppose $Y_1\le \cdots \le Y_{n'}$.  Then certain quantiles of $G$ will satisfy our needs with probability $1-\beta$.  More precisely:

\begin{lemma}
\label{lm:obs_diff}
Given $\beta$, for any $D\in\mathcal{P}^n$, if $n>c_1\log(1/\beta)$,
where $c_1$ is a universal constant, then with probability at least $1-\beta$, we have,
\[\varphi\left(\frac{1}{16}\right)\leq Y_{\frac{5n'}{32}} \]
and
\[Y_{\frac{7n'}{32}}\leq \mathrm{IQR}.\]
\end{lemma}

\begin{proof}
First, by Lemma~\ref{lm:lower_bound_diff}, we have
\[\E\left[\left| G\cap\left[0,\varphi\left(\frac{1}{16}\right)\right]\right|\right]\leq \frac{n'}{8}.\]
And similarly,
\[\E\left[\left| G\cap\left[0,\mathrm{IQR}\right]\right|\right]\geq \frac{n'}{4}.\]
Then both parts of the conclusion follow from Chernoff's inequality with a sufficiently large $c_1$.
\end{proof}

Therefore, we can find a quantile between $Y_{\frac{5n'}{32}}$ and $Y_{\frac{7n'}{32}}$, say $Y_{\frac{3n'}{16}}$, as $\overline{\mathrm{IQR}}$. However, we cannot use $\mathtt{InfiniteDomainQuantile}$ here as we have not discretized $\mathbb{R}$ yet.  To get out of this circular dependency, we obverse that we do not need a $\tilde{Y}_{\frac{3n'}{16}}$ with a small rank error; instead, a rough constant-factor approximation will do.  Thus, the idea is to run two instances of SVT, one with increasing thresholds and one with decreasing thresholds, as detailed in Algorithm \ref{alg:bucket_find}.

\begin{algorithm}
\LinesNumbered 
\label{alg:bucket_find}
\caption{$\mathtt{EstimateIQRLowerBound}$.}
\KwIn{$D$, $\varepsilon$, $\beta$}
$n'=\frac{n}{2}$\;
Construct $G$ from $D$\;
$\tilde{i} = \mathrm{SVT}\left(\frac{3n'}{16},\frac{\varepsilon}{2},  \mathrm{Count}(G,2^0), \mathrm{Count}(G,2^1),\mathrm{Count}(G,2^2),\dots\right)$\;
$\tilde{j} = \mathrm{SVT}\left(-\frac{3n'}{16},\frac{\varepsilon}{2}, - \mathrm{Count}(G,2^0), -\mathrm{Count}(G,2^{-1}),-\mathrm{Count}(G,2^{-2}),\dots\right)$\;
\uIf{$\tilde{i}>1$}{
    $\overline{\mathrm{IQR}} = 2^{\tilde{i}-2}$\;
}
\Else{
    $\overline{\mathrm{IQR}} = 2^{-\tilde{j}}$\;
}
\Return $\overline{\mathrm{IQR}}$\;
\end{algorithm}

The privacy of  $\mathtt{EstimateIQRLowerBound}$ is straightforward; we analyze its utility below:

\begin{theorem}
\label{th:err_bucket}
Given $\varepsilon$, $\beta$, for any $D\sim \mathcal{P}^n$, if
\[n>\frac{c_1}{\varepsilon}\log\log{1 \over \varphi(1/16)}+\frac{c_2}{\varepsilon}\log\log\left(\mathrm{IQR}\right)+\frac{c_3}{\varepsilon}\log(1/\beta),\]
where $c_1$, $c_2$, $c_3$ are universal constants, then with probability at least $1-\beta$, $\mathtt{EstimateIQRLowerBound}$ returns an $\overline{\mathrm{IQR}}$ such that,
\[\frac{1}{4}\cdot \varphi\left(\frac{1}{16}\right) \leq \overline{\mathrm{IQR}} \leq \mathrm{IQR}.\]
\end{theorem}

\begin{proof}
First, by Lemma~\ref{lm:obs_diff} and setting $c_3$ sufficiently large, with probability at least $1-\frac{\beta}{5}$,
\begin{equation}
\label{eq:th:err_bucket_1}
\varphi\left(\frac{1}{16}\right)\leq Y_{\frac{5n'}{32}}
\end{equation}
and
\begin{equation}
\label{eq:th:err_bucket_2}
Y_{\frac{7n'}{32}}\leq \mathrm{IQR}.
\end{equation}

Consider the following three cases: (1) $Y_{\frac{7n'}{32}}\leq 1$, (2) $Y_{\frac{5n'}{32}}\geq 1$, and (3) $Y_{\frac{5n'}{32}}<1<Y_{\frac{7n'}{32}}$. For case (1), by Lemma~\ref{lm:simple_obs_SVT} and by setting a large $c_3$, we have with probability at least $1-\frac{\beta}{5}$, the first $\mathrm{SVT}$ instance will stop at the first query. i.e., $\tilde{i} = 1$.
Then, by Lemma~\ref{lm:simple_obs_SVT}, (\ref{eq:th:err_bucket_1}), and setting $c_1$ and $c_3$ large enough, we have with probability at least $1-\frac{2\beta}{5}$, the second $\mathrm{SVT}$ instance will stop at the $\tilde{j}$-th query such that
\begin{equation*}
\frac{1}{2}\cdot Y_{\frac{5n'}{32}}\leq 2^{-\tilde{j}+1} \leq Y_{\frac{7n'}{32}}, \end{equation*}
which, by (\ref{eq:th:err_bucket_1}) and (\ref{eq:th:err_bucket_2}), implies 
\[\frac{1}{4}\cdot \varphi\left(\frac{1}{16}\right)\leq 2^{-\tilde{j}} \leq \frac{1}{2}\cdot \mathrm{IQR}.\]
In this case, the algorithm returns $2^{-\tilde{j}+1}$, thus the conclusion follows.

For case (2), similarly, we can derive, with probability at least $1-\frac{2\beta}{5}$, 
\begin{equation*}
Y_{\frac{5n'}{32}}\leq 2^{\tilde{i}-1} \leq 2\cdot Y_{\frac{7n'}{32}}\Rightarrow \frac{1}{2}\cdot \varphi\left(\frac{1}{16}\right)\leq 2^{\tilde{i}-2} \leq \mathrm{IQR},
\end{equation*}
and with probability at least $1-\frac{\beta}{5}$, $\tilde{j} = 1$. In this case, if $\tilde{i}\neq 1$, the algorithm will return $2^{\tilde{i}-2}$, which implies the target bound. If $\tilde{i}=1$, the target bound still holds since $2^{\tilde{i}-2} = 2^{-\tilde{j}}$.

For case (3), we have with probability at least $1-\frac{4\beta}{5}$, 
\[\frac{1}{4}\cdot \varphi\left(\frac{1}{16}\right)\leq 2^{-\tilde{j}} \leq \frac{1}{2}\cdot \mathrm{IQR},\]
and
\[\frac{1}{2}\cdot \varphi\left(\frac{1}{16}\right)\leq 2^{\tilde{i}-2} \leq \mathrm{IQR}.\]
Thus the conclusion holds no matter whether $\tilde{i}=1$ or not.
\end{proof}

\subsection{General Algorithm and Error Analysis}
\label{sec:mean_alg}

We mentioned that directly invoking $\mathtt{InfiniteDomainMean}$ over $D$, even with a good bucket size, results sub-optimal errors in the statistical setting with respect to the dependency on $\varepsilon$.  Here we give an intuitive explanation.  Recall that in $\mathtt{InfiniteDomainMean}$, we find a privatized range $\tilde{\mathcal{R}}(D)$ and use it with the clipped mean estimator. The error comes from two sources: (1) There are $\tilde{O}(1/\varepsilon)$ clipped outliers, each contributing $\gamma(D)/n$ bias.  (2) The Laplace noise is proportional to $|\tilde{\mathcal{R}}(D)|/(\varepsilon n) = O(\gamma(D)/(\varepsilon n)$.  One should thus match the two parts of errors for an optimal overall error bound.  In the empirical setting, as $D$ is arbitrary, simply using $\gamma(D)/n$ as an upper bound on the bias from clipping each outlier is already the best one can do.  In the statistical setting, however, since $D$ is an i.i.d.\ sample, this upper bound is too pessimistic.

Therefore, in the statistical setting, we try to use a tighter $\tilde{\mathcal{R}}(D)$ to perform more aggressive clipping.  The idea is to sub-sample $m$ elements from $D$ and obtain a privatized range on the sample $D'$, denoted $\tilde{\mathcal{R}}(D')$.  A smaller $m$ corresponds to more aggressive clipping, which increases the bias but reduces the noise. The optimal choice of $m$ will depend on $\mathcal{P}$, which is not possible for a universal estimator. Fortunately and somehow amazingly, $m=\varepsilon n$ turns out to be a choice that is good enough, and here is the intuition: By Theorem~\ref{th:sampling_amplification}, the privacy budget on finding $\tilde{\mathcal{R}}(D')$ can be amplified to $\varepsilon' \approx \varepsilon n/m$. Therefore, there are $\tilde{O}(1/\varepsilon') =\tilde{O}(m/(\varepsilon n))$ outliers in $D'$ outside $\tilde{\mathcal{R}}(D')$.  However, there is essentially no room for improvement when the number of outliers in $D'$ is less than $1$, i.e., it is sufficient to set $m\ge \varepsilon n$. When $m\ge \varepsilon n$, the number of outliers in $D$ is roughly $\tilde{O}(m/(\varepsilon n)) \cdot n/m = \tilde{O}(1/\varepsilon n)$, which is fixed, while a smaller $m$ reduces $|\tilde{\mathcal{R}}(D')|$.

\begin{algorithm}
\LinesNumbered 
\label{alg:mean}
\caption{$\mathtt{EstimateMean}$.}
\KwIn{$D$, $\varepsilon$, $\beta$}
$\overline{\mathrm{IQR}}=\mathtt{EstimateIQRLowerBound}(D,\frac{\varepsilon}{8},\frac{\beta}{9})$\;
Let $D'$ be a sample of $\varepsilon n$ values from $D$\;
$\varepsilon' = \log\left(\frac{e^{\varepsilon}-1}{\varepsilon}+1\right)$\;
$\tilde{\mathcal{R}}(D') = \mathtt{InfiniteDomainRange}(D',\frac{3\varepsilon'}{4},\frac{\beta}{9})$ with $b=\overline{\mathrm{IQR}}$\;
$\tilde{\mu} =\mathrm{ClippedMean}\left(D,\tilde{\mathcal{R}}(D')\right) +\mathrm{Lap}\left(8|\tilde{\mathcal{R}}(D')|/(\varepsilon n)\right)$\;
\Return $\tilde{\mu}$\;
\end{algorithm}

With the intuition above, we present our statistical mean estimator, as shown in Algorithm~\ref{alg:mean}.  Its privacy follows from Theorem~\ref{th:sampling_amplification} and basic composition.  Before analyzing its error, we first state a standard result relating $\mathcal{P}$ with its truncated version:

\begin{lemma}
\label{lm:obs_truncated_distribution}
Let $X\sim \mathcal{P}$ and $\xi\geq 0$, and let $\bar{X}$ be the following random variable:
\[\bar{X} = \begin{cases}
\mu - \xi, & \text{if }X<\mu - \xi;
\\
X, & \text{if }\mu - \xi \leq X\leq \mu +\xi;
\\
\mu + \xi, & \text{if }X>\mu +\xi.
\end{cases}\]
Let $\bar{\mu}$ and $\bar{\sigma}^2$ denote the mean and variance of $\bar{X}$. Then,
\[\bar{\sigma} \leq \sigma,\]
and
\[\mu-\bar{\mu} = \E[X<\mu-\xi]+\E[X>\mu+\xi].\]
\end{lemma}

We are now ready to analyze the error of $\mathtt{EstimateMean}$. 

\begin{theorem}
\label{th:err_mean}
Given $\varepsilon$, $\beta$, for any $D\sim \mathcal{P}^n$, if 
\begin{align*}
n>&\frac{c_1}{\varepsilon}\log\log{1 \over \varphi(1/16)}+\frac{c_2}{\varepsilon}\log\log\left(\mathrm{IQR}\right)+\frac{c_3}{\varepsilon}\log{1 \over \beta} +\frac{c_4}{\varepsilon}\log{|\mu|+\sigma+\gamma(\varepsilon n,\beta/9) \over \varphi(1/16)},
\end{align*}
where $c_1$, $c_2$, $c_3$, and $c_4$ are universal constants, then with probability at least $1-\beta$, $\mathtt{EstimateMean}$ returns a value $\tilde{\mu}$ such that
\begin{align}
\nonumber
|\mu-\tilde{\mu}|
=& O\left(\min_{\xi\geq10\cdot\gamma\left(\varepsilon n,\frac{\beta}{9}\right)+2\sigma} \left(\left|\E\left[X<\mu-\xi\right]+\E\left[X>\mu+\xi\right]\right|+\frac{\xi}{\varepsilon n}\log\left({1\over \beta}\log{\gamma(\varepsilon n,\beta/9) \over \varphi(1/16)}\right)\right)+\sigma\sqrt{\frac{\log(1/\beta)}{n}}\right).
\label{eq:th:err_mean_1}
\end{align}
\end{theorem}

We first explain each term in the theorem before presenting its proof.  The first two terms in the requirement of $n$ are from finding the bucket size, and the last one is for estimating $\tilde{\mathcal{R}}(D')$. In the error bound, all the terms in the $\min_{\xi}$ are due to privacy, while the last term is the sampling error.  We would like to emphasize that although the requirement on $n$ and the error bound depend on $\mathcal{P}$ (they have to), the algorithm does not need any \textit{a priori} assumptions on $\mathcal{P}$.  Furthermore, some of the dependencies can be improved if certain assumptions are made on $\mathcal{P}$.  For instance, if $\sigma_{\mathrm{min}}$ is given, then there is no need to find a bucket size and the first two terms in the requirement on $n$ will disappear, while the $\varphi\left(\frac{1}{16}\right)$ in both the requirement on $n$ and the error bound will be replaced by $\sigma_{\min}$. 

\begin{proof}
For convenience, denote $D' = \{X_1',\dots,X_{\varepsilon n}'\}$, where $X'_1 \le \cdots \le X_{\varepsilon n}'$. 

First, by Theorem~\ref{th:err_bucket} and setting $c_1$, $c_2$, and $c_3$ large enough, we have with probability at least $1-\frac{\beta}{9}$, 
\begin{equation}
\label{eq:th:err_mean_1}
\frac{1}{4}\cdot \varphi\left(\frac{1}{16}\right) \leq \bucket\leq \mathrm{IQR}.
\end{equation}

By the definition of $\gamma(\varepsilon n,\beta/9)$, we have with probability at least $1-\frac{\beta}{9}$,
\begin{equation}
\label{eq:th:err_mean_2}
\gamma(D')\leq \gamma\left(\varepsilon n,\frac{\beta}{9}\right).
\end{equation}

Consider any $X\in\mathcal{P}$. Since with probability at least $\frac{3}{4}$, $X\in[\mu-2\sigma,\mu+2\sigma]$, by setting $c_3$ large enough, with probablity at least $1-\frac{\beta}{9}$, we have
\begin{equation}
\label{eq:th:err_mean_3}
\left|D'\cap [\mu-2\sigma,\mu+2\sigma]\right| = \Omega(\varepsilon n).
\end{equation}
Combining with (\ref{eq:th:err_mean_2}), we have
\begin{equation}
\label{eq:th:err_mean_4}
\mathrm{rad}(D')\leq |\mu|+2\sigma+\mathrm{rad}\left(\varepsilon n,\frac{\beta}{9}\right).
\end{equation}

Recall from Section~\ref{sec:notation} that $\mathrm{IQR}\leq \gamma\left(\log_{\frac{4}{3}}(2/\beta),\beta\right)$. Thus by setting $c_3$ large enough, we have
\begin{equation}
\label{eq:th:err_mean_5}
\mathrm{IQR}\leq \gamma\left(\varepsilon n,\frac{\beta}{9}\right).
\end{equation}

Then, by (\ref{eq:th:err_mean_4}), Theorem~\ref{th:err_unbounded_range} and setting $c_3$ and $c_4$ large enough, we have with probability at least $1-\frac{\beta}{9}$,
\begin{equation}
\label{eq:th:err_mean_6}
|\tilde{\mathcal{R}}(D')|\leq 4\cdot \gamma(D')+6\bucket \leq 10\cdot \gamma\left(\varepsilon n,\frac{\beta}{9}\right),
\end{equation}
where the last inequality is by (\ref{eq:th:err_mean_1}), (\ref{eq:th:err_mean_2}), and (\ref{eq:th:err_mean_5}); and
\begin{align}
\left|D'\cap \overline{\tilde{\mathcal{R}}(D')}\right| = O\left(\frac{1}{\varepsilon'}\log\left({1\over \beta}\log {\gamma(D') \over \bucket}\right)\right)
\leq  c_5 \log\left({1\over \beta} \log{\gamma(\varepsilon n, \beta/9) \over \varphi(1/16)}\right),
\label{eq:th:err_mean_7}
\end{align}
for some constant $c_5$ large enough, where the second inequality is by (\ref{eq:th:err_mean_1}) and (\ref{eq:th:err_mean_2}).

Combining (\ref{eq:th:err_mean_3}) and (\ref{eq:th:err_mean_7}) and setting $c_3$ and $c_4$ large enough, we have
\begin{equation}
\label{eq:th:err_mean_8}
\left|D'\cap [\mu-2\sigma, \mu+2\sigma]\cap \tilde{\mathcal{R}}(D')\right| \ge 1,
\end{equation}
so $[\mu-2\sigma, \mu+2\sigma]$ and $\tilde{\mathcal{R}}(D')$ overlap.

Next, define $\zeta = \frac{2c_5}{\varepsilon n}\log\left({10 \over \beta} \log{\gamma(\varepsilon n,\beta/9) \over \varphi(1/16)}\right)$. By setting $c_3$ and $c_4$ large enough, we can ensure $\zeta<0.5$. Now we consider the interval $(-\infty,F^{-1}(\zeta)]$. By Chernoff's inequality, with probability at least $1-\frac{\beta}{9}$,
\begin{equation}
\label{eq:th:err_mean_9}
\left|D'\cap (-\infty,F^{-1}(\zeta)]\right|\geq c_5 \log\left({9\over \beta}\log{\gamma(\varepsilon n,\beta/9) \over \varphi(1/16)}\right).
\end{equation}
Similarly, with probability at least $1-\frac{\beta}{9}$,
\begin{equation}
\label{eq:th:err_mean_10}
\left|D'\cap [F^{-1}(1-\zeta),\infty)\right|\geq c_5 \log\left({9\over \beta}\log{\gamma(\varepsilon n,\beta/9) \over \varphi(1/16)}\right).
\end{equation}
Combining (\ref{eq:th:err_mean_7}), (\ref{eq:th:err_mean_9}), (\ref{eq:th:err_mean_10}), we have,
\[\left[F^{-1}(\zeta),F^{-1}(1-\zeta)\right]\subseteq\tilde{R}(D').\]
Furthermore, by Chernoff's inequality, with probability at least $1-\frac{\beta}{9}$,
\begin{align}
&\left|D\cap \left[F^{-1}(\zeta),F^{-1}(1-\zeta)\right]\right| \geq n- \frac{8c_5}{\varepsilon}\log\left( {10\over \beta}\log{\gamma(\varepsilon n,\beta/9) \over \varphi(1/16)}\right)
\nonumber
\\
\nonumber
\Rightarrow \qquad&
\left|D\cap \tilde{\mathcal{R}}(D')\right| \geq n- \frac{8c_5}{\varepsilon}\log\left( {10\over \beta}\log{\gamma(\varepsilon n,\beta/9) \over \varphi(1/16)}\right)
\\
\label{eq:th:err_mean_11}
\Rightarrow \qquad &\left|D\cap \overline{\tilde{\mathcal{R}}(D')}\right|=O\left(\frac{1}{\varepsilon}\log\left( {1\over \beta}\log{\gamma(\varepsilon n,\beta/9) \over \varphi(1/16)}\right) \right).
\end{align}

Now, we start analyzing the error. Given any $\xi\geq 10\cdot\gamma\left(\varepsilon n,\frac{\beta}{9}\right)+2\sigma$, let $\bar{\mathcal{P}}$ be the distribution after truncating $\mathcal{P}$ into the interval $[\mu-\xi,\mu+\xi]$. Denote the mean and variance of $\bar{\mathcal{P}}$ as $\bar{\mu}$ and $\bar{\sigma}^2$. By Lemma~\ref{lm:obs_truncated_distribution}, we have
\begin{equation*}
\bar{\sigma}\leq \sigma,
\end{equation*}
and 
\begin{equation*}
\mu-\bar{\mu} = \E[X<\mu-\xi]+\E[X>\mu+\xi].
\end{equation*}

Since $|\mu-\tilde{\mu}|\leq |\mu-\bar{\mu}|+|\bar{\mu}-\tilde{\mu}|$, it remains to bound $|\bar{\mu}-\tilde{\mu}|$.
Denote $\bar{D} = \mathrm{Clip}(D, [\mu-\xi,\mu+\xi])$.  The error can be broken down into two parts:
\begin{align*}
|\bar{\mu}-\tilde{\mu}| \leq \left|\bar{\mu}-\mu(\bar{D})\right|+\left|\mu(\bar{D})-\tilde{\mu}\right|.
\end{align*}

The first part is the sampling error. By Bernstein's inequality, we have with probability at least $1-\frac{\beta}{9}$,
\begin{equation*}
\left|\bar{\mu}-\mu(\bar{D})\right| = O\left(\bar{\sigma}\sqrt{\frac{\log(1/\beta)}{n}}+\frac{\xi\log(1/\beta)}{n}\right)= O\left(\sigma\sqrt{\log(1/\beta) \over n}+\frac{\xi\log(1/\beta)}{n}\right).
\end{equation*}

The second part can be further divided into two sub-parts:
\begin{align*}
\left|\mu(\bar{D})-\tilde{\mu}\right| = & \left|\mu(\bar{D})-\left(\mathrm{Lap}\left(\frac{8|\tilde{\mathcal{R}}(D')|}{\varepsilon n}\right)+\mathrm{ClippedMean}\left(D,\tilde{\mathcal{R}}(D')\right)\right) \right|
\\
\leq & \left|\mathrm{Lap}\left(\frac{8|\tilde{\mathcal{R}}(D')|}{\varepsilon n}\right)\right|+\left|\mu(\bar{D})-\mathrm{ClippedMean}\left(D,\tilde{\mathcal{R}}(D')\right)\right|
\\
=& \left|\mathrm{Lap}\left(\frac{8|\tilde{\mathcal{R}}(D')|}{\varepsilon n}\right)\right|+ \left|\mu(\bar{D})-\mathrm{ClippedMean}\left(\bar{D},\tilde{\mathcal{R}}(D')\right)\right|.
\end{align*}
The last equality is because, by (\ref{eq:th:err_mean_6}), (\ref{eq:th:err_mean_8}) together with $\xi\geq 10 \cdot \gamma\left(\varepsilon n,\frac{\beta}{9}\right)+2\sigma$, we have $\tilde{\mathcal{R}}(D')\subseteq [\mu-\xi,\mu+\xi]$, thus 
\[\mathrm{ClippedMean}\left(D,\tilde{\mathcal{R}}(D')\right)=\mathrm{ClippedMean}\left(\bar{D},\tilde{\mathcal{R}}(D')\right).\]

For the first sub-part, with probability at least $1-\frac{\beta}{9}$,
\begin{equation*}
\left|\mathrm{Lap}\left(\frac{8|\tilde{\mathcal{R}}(D')|}{\varepsilon n}\right)\right|\leq \frac{8\log(9/\beta)|\tilde{\mathcal{R}}(D')|}{\varepsilon n}= O\left(\frac{\log(1/\beta)}{\varepsilon n}\gamma\left(\varepsilon n,\frac{\beta}{9}\right)\right).
\end{equation*}
The first inequality is because of the tail bound of the Laplace distribution while the second one is by (\ref{eq:th:err_mean_6}).

The second sub-part is because we clip some outliers in $\bar{D}$ out of $\tilde{\mathcal{R}}(D')$. By (\ref{eq:th:err_mean_11}) and the fact that each outlier will contribute a bias at most $2\xi/n$. Therefore,
\begin{equation*}
\left|\mu(\bar{D})-\mathrm{ClippedMean}\left(\bar{D},\tilde{\mathcal{R}}(D')\right)\right| = O\left(\frac{\xi}{\varepsilon n} \log\left({1\over \beta} \log{\gamma(\varepsilon n,\beta/9)  \over \varphi(1/16)}\right)\right).
\end{equation*}
\end{proof}

\subsection{Error Bounds for Specific Distribution Families}
\label{sec:mean_error_specific}

To facilitate the comparison with prior work, below we derive simplified (and possibly looser) versions of Theorem \ref{th:err_mean} for certain distribution families.  These simplified bounds can be easily rewritten into the sample complexity results stated in Section \ref{sec:intro}.  We also set $\beta$ as $\frac{1}{3}$.

\paragraph{Gaussian distributions.}
For a Gaussian $\mathcal{P}$, we have $\varphi(\beta) = \Theta(\sigma)$, $\mathrm{IQR} = \Theta(\sigma)$, and $\gamma(\varepsilon n,\beta/9)  = O\left(\sigma\sqrt{\log(\varepsilon n)}\right)$ by the standard Gaussian tail bound. In addition, due to its symmetry, $\E[X<\mu-\xi]+\E[X>\mu+\xi] = 0$ for any $\xi$.  Fixing $\xi = c\sigma \sqrt{\log(\varepsilon n)}$ for some large constant $c$, Theorem \ref{th:err_mean} simplifies into:

\begin{theorem}
\label{th:err_mean_gaussian}
Given $\varepsilon$, $\beta$, for any $D\sim \mathcal{P}^n$, where $\mathcal{P}$ is a Gaussian distribution, if 
\[n> \frac{c_1}{\varepsilon} \log\log\sigma + \frac{c_2}{\varepsilon}\log\log {1\over \sigma}+\frac{c_3}{\varepsilon}\log{|\mu| \over \sigma},\]
where $c_1$, $c_2$, $c_3$ are universal constants, then
\begin{align*}
\mathrm{Err}(\mathtt{EstimateMean},D)=& O\left(\frac{\sigma}{\sqrt{n}}+\frac{\sigma}{\varepsilon n}\log\log\log(\varepsilon n)\sqrt{\log(\varepsilon n)}\right).
\end{align*}
\end{theorem}

\paragraph{Heavy-tailed distributions.}
Now, we consider the case where $\mathcal{P}$ has a bounded $k$-th central moment $\mu_{k}$. Note that $\sigma\leq \mu_k^{1/k}$.  In addition, we can also bound $\gamma(m,\beta)$ in terms of $\mu_k$:

\begin{lemma}
\label{lm:obs_m_range}
For any $m,\beta$, and $k\geq 2$,
\[\gamma(m,\beta)\leq 2 \left(\frac{m \mu_k}{\beta}\right)^{1/k}.\]
\end{lemma}

\begin{proof}
By Chebyshev's inequality, given an $X\sim \mathcal{P}$, with probability at least $1-\frac{\beta}{m}$,
\[|X-\mu|\leq \left(\frac{m \mu_k}{\beta}\right)^{1/k}.\]
Then applying the union bound over $m$ such random variables yields the lemma.
\end{proof}

Plugging these bounds into Theorem~\ref{th:err_mean} and setting $\xi = c\cdot \left(\varepsilon n \mu_k\right)^{1/k}$ for some large constant $c$, the requirement on $n$ becomes
\[n> \frac{c_1}{\varepsilon}\log\log{1 \over \varphi(1/16)}+\frac{c_2}{\varepsilon}\log\log\left(\mathrm{IQR}\right)+\frac{c_3}{\varepsilon}\log{|\mu|+(\varepsilon \mu_k)^{1/k} \over \varphi(1/16)},\]
and the error bound changes to 
\begin{align}
\nonumber
\mathrm{Err}(\mathrm{\mathtt{EstimateMean}},D)
=& O\left(\frac{\sigma}{\sqrt{n}} + \frac{\mu_k^{1/k}}{(\varepsilon n)^{1-1/k}} \log\log{(\varepsilon n\mu_k)^{1/k} \over \varphi(1/16)}\right.
\\
&\left.+\left|\E\left[X<\mu-c\cdot\left(\varepsilon n \mu_k\right)^{1/k}\right]+\E\left[X>\mu+c\cdot\left(\varepsilon n \mu_k\right)^{1/k}\right]\right|\right).
\label{eq:mean_heavy_tail_distribution}
\end{align}

Now, we further analyze the last term in (\ref{eq:mean_heavy_tail_distribution}). We first derive a lemma similar to the one in~\cite{kamath2020private}:

\begin{lemma}
\label{lm:obs_k_moment}
Let $\mathcal{P}$ be a distribution with a bounded $\mu_k$. Given $\xi$ and $t$ such that $\xi\geq 2\left(\mu_k/t\right)^{1/(k-1)}$, we have 
\[\left|\E[X<\mu-\xi]+\E[X>\mu+\xi]\right|\leq t.\]
\end{lemma}

\begin{proof}
It suffices to bound $\left|\E[X<\mu-\xi]\right|$ and $\left|\E[X>\mu+\xi]\right|$, respectively.  We only consider the former; the latter is symmetric. 
\begin{align}
\nonumber
\left|\E[X<\mu-\xi]\right|= & \left|\E_{X\sim \mathcal{P}}[\left(X-(\mu-\xi)\right)\mathbb{I}(X<\mu-\xi)]\right|
\\
\nonumber
\leq & \left|\E_{X\sim \mathcal{P}}[(X-\mu)\mathbb{I}(X<\mu-\xi)]\right|
\\
\label{eq:lm:obs_k_moment_1}
\leq & \left(\E_{X\sim \mathcal{P}}\left[|X-\mu|^k\right]\right)^{1/k} \left(\Pr_{X\sim \mathcal{P}}[X\leq \mu-\xi]\right)^{1-1/k}
\\
\nonumber
\leq & \mu_k^{1/k} \left(\Pr_{X\sim \mathcal{P}}[|X-\mu|^k\leq \xi^k]\right)^{1-1/k}
\\
\leq & \mu_k^{1/k}
\nonumber \left(\frac{\mu_k}{\xi^k}\right)^{1-1/k}
\\
\leq & \frac{t}{2},
\nonumber
\end{align}
where (\ref{eq:lm:obs_k_moment_1}) follows from H\"{o}lder's inequality.
\end{proof}

By setting $\xi=c\cdot\left(\varepsilon n \mu_k\right)^{1/k}$ for $c\ge 2$ and $t = \frac{\mu_k^{1/k}}{(\varepsilon n)^{1-1/k}}$, we have
\[\left|\E\left[X<\mu-c\cdot\left(\varepsilon n \mu_k\right)^{1/k}\right]+\E\left[X>\mu+c\cdot\left(\varepsilon n \mu_k\right)^{1/k}\right]\right| \leq \frac{\mu_k^{1/k}}{(\varepsilon n)^{1-1/k}}.\]
Plugging this bound into (\ref{eq:mean_heavy_tail_distribution}), we obtain:

\begin{theorem}
\label{th:err_mean_heavy}
Given $\varepsilon$, $\beta$, for any $D\sim\mathcal{P}^n$ and any $k$ if 
\[n> \frac{c_1}{\varepsilon}\log\log{1\over\varphi(1/16)}+\frac{c_2}{\varepsilon}\log\log\left(\mathrm{IQR}\right)+\frac{c_3}{\varepsilon}\log{|\mu|+(\varepsilon \mu_k)^{1/k} \over \varphi(1/16)},\]
where $c_1$, $c_2$, $c_3$ are universal constants, then
\[\mathrm{Err}(\mathrm{\mathtt{EstimateMean}},D)=O\left(\frac{\sigma}{\sqrt{n}} + \frac{\mu_k^{1/k}}{(\varepsilon n)^{1-1/k}} \log\log{(\varepsilon  n\mu_k)^{1/k} \over \varphi(1/16)}\right).\]
\end{theorem}

\section{Statistical Variance Estimation}
\label{sec:var}

\subsection{General Algorithm and Error Analysis}
\label{sec:var_alg}

We first use a standard technique to reduce the variance estimation problem to mean estimation. Define a random variable $Z=(X-X')^2$, where $X,X'\sim \mathcal{P}$. Note that $Z$ has a non-negative domain. Let $\mathcal{T}$ be the distribution of $Z$.  We can relate the statistical parameters of $\mathcal{T}$ to those of $\mathcal{P}$ (statistical parameters without a subscript refer to $\mathcal{P}$):
\begin{equation}
\label{eq:mean_T}
\mu_{\mathcal{T}} = \E_{Z\sim \mathcal{T}}[Z] = E_{X,X'\sim \mathcal{P}}[(X-X')^2]= 2\sigma^2,
\end{equation}
\begin{align}
\sigma_{\mathcal{T}}^2 =& \E_{Z\sim \mathcal{T}}[Z^2]-\E^2_{Z\sim \mathcal{T}}[Z]
\nonumber
\\
\nonumber
= & \E_{X,X'\sim \mathcal{P}}[(X-X')^4] - 4\sigma^4
\\
\nonumber
= & 2\mu_{4} - 4\sigma^4
\\
\leq & 2\mu_{4}.
\label{eq:sigma_T}
\end{align}

We can also derive a connection between the statistical width of $\mathcal{T}$ and that of $\mathcal{P}$:

\begin{lemma}
\label{lm:connection_statistical_width}
For any $m>1$, $\beta$,
\[\gamma_{\mathcal{T}}(m,\beta)\leq \left(\gamma(2m,\beta)\right)^2.\]
\end{lemma}

\begin{proof}
Any sample $H\in \mathcal{T}^m$ corresponds to a sample $D\in\mathcal{P}^{2m}$. By definition, with probability at least $1-\beta$,
\[|X_1-X_{2m}|\leq \gamma(2m,\beta).\]
That is, for any $i,j\in[2m]$,
\[(X_i-X_j)^2\leq  \left(\gamma(2m,\beta)\right)^2,\]
which further means, for any $k\in [m]$
\[Z_k\leq \left(\gamma(2m,\beta)\right)^2.\]
Since $Z\in \mathcal{T}$ has the non-negative domain, we get the lemma.
\end{proof}

By (\ref{eq:mean_T}), we can estimate $\sigma^2$ by estimating $\mu_{\mathcal{T}}$. Thus, we randomly group the elements in $D$ into pairs $(X,X')$ and let $H=\{Z_1,Z_2,\cdots,Z_{n'}\}$, where $n'=n/2$ and $Z_i = (X-X')^2$ for each pair. Note that $H$ is a random sample drawn from $\mathcal{T}$.  Then, we estimate $\mu_{\mathcal{T}}$ with a similar procedure as before. We first find $\overline{\mathrm{IQR}}$ on $D$ but use $\overline{\mathrm{IQR}}^2$ as the bucket size. Next, we randomly sample $\varepsilon n'$ elements from $H$ to get $H'$. When we try to get the privatized range for $H'$, instead of using $\mathtt{InfiniteDomainRange}$, we simply use $\mathtt{InfiniteDomainRadius}$ to get the a range $[0, \widetilde{\mathrm{rad}}(H)]$.  This is because the sampling error will be proportional to $\sqrt{\mu_4} \ge \sigma^2 = {1\over 2} \mu_{\mathcal{T}}$, there is no need to find the location of the range.  This removes a term in the requirement on $n$. Finally, we use the clipped mean estimator with the range $[0, \widetilde{\mathrm{rad}}(H)]$. The details are shown in Algorithm~\ref{alg:var}.

\begin{algorithm}
\LinesNumbered 
\label{alg:var}
\caption{$\mathtt{EstimateVariance}$.}
\KwIn{$D$, $\varepsilon$, $\beta$}
$\overline{\mathrm{IQR}}=\mathtt{EstimateIQRLowerBound}(D,\frac{\varepsilon}{8}, \frac{\beta}{7})$\;
$n'=\frac{n}{2}$\;
Construct $H$ from $D$\;
Let $H'$ be a sample of $\varepsilon n'$ values from $H$\;
$\varepsilon' = \log\left(\frac{e^{\varepsilon}-1}{\varepsilon}+1\right)$\;
$\widetilde{\mathrm{rad}}(H')= \mathtt{InfiniteDomainRadius}(H',\frac{3\varepsilon'}{4},\frac{\beta}{7})$ with $b=\overline{\mathrm{IQR}}^2$\;
$\tilde{\sigma}^2 =\frac{1}{2}\left(\mathrm{ClippedMean}\left(H,[0, \widetilde{\mathrm{rad}}(H')]\right) +\mathrm{Lap}\left(8\cdot \widetilde{\mathrm{rad}}(H')/(\varepsilon n)\right)\right)$\;
\Return $\tilde{\sigma}^2$\;
\end{algorithm}

The privacy of $\mathtt{EstimateVariance}$ is straightforward. We analyze its utility below:

\begin{theorem}
\label{th:err_var}
Given $\varepsilon$, $\beta$, for any $D\sim \mathcal{P}^n$, if 
\[n>\frac{c_1}{\varepsilon}\log\log{1 \over \varphi\left(1/16\right)}+\frac{c_2}{\varepsilon}\log\log\left(\mathrm{IQR}\right)+\frac{c_3}{\varepsilon}\log{1\over\beta},\]
where $c_1$, $c_2$, $c_3$ are universal constants, then with probability at least $1-\beta$, $\mathtt{EstimateVariance}$ returns a $\tilde{\sigma}^2$ such that
\begin{align*}
|\sigma^2-\tilde{\sigma}^2|=&O\left(\min_{\xi\geq 5\cdot \left(\gamma(\varepsilon n,\frac{\beta}{7})\right)^2+2\sigma^2}\left(\left|\E\left[Z>2\sigma^2+\xi\right]\right|+\frac{\xi}{\varepsilon n} \log\left({1\over \beta}\log{\gamma\left(\varepsilon n,\beta/7\right)\over \varphi\left(1/16\right)}\right) \right)
+ \sqrt{\frac{\mu_4\log(1/\beta)}{n}}\right).
\end{align*}
\end{theorem}

\begin{proof}
For convenience, denote $H' = \{Z_1',\dots, Z_{\varepsilon n'}'\}$.

First, by Theorem~\ref{th:err_bucket} and setting $c_1$, $c_2$, $c_3$ sufficiently large, with probability at least $1-\frac{\beta}{7}$, we have
\begin{equation}
\label{eq:th:err_var_1}
\frac{1}{16} \cdot \varphi\left(\frac{1}{16}\right)^2 \leq \bucket\leq \mathrm{IQR}^2.
\end{equation}

By the definition of $\gamma_{\mathcal{T}}$, with probability at least $1-\frac{\beta}{7}$, 
\begin{equation}
\label{eq:th:err_var_2}
\mathrm{rad}(H') = Z'_{\varepsilon n'}\leq  \gamma_{\mathcal{T}}(\varepsilon n',\beta/7) \leq \left(\gamma(\varepsilon n,\beta/7)\right)^2.
\end{equation}

By setting a large $c_3$ in the condition of $n$, we have
\begin{equation}
\label{eq:th:err_var_3}
\mathrm{IQR}\leq \gamma\left(\varepsilon n,\frac{\beta}{7}\right).
\end{equation}

Then, by Theorem~\ref{th:err_unbounded_radius}, we have with probability at least $1-\frac{\beta}{7}$,
\begin{equation*}
\label{eq:th:err_var_4}
\widetilde{\mathrm{rad}}(H')\leq 2\cdot \mathrm{rad}(H')+3\bucket \leq 5 \cdot \left(\gamma(\varepsilon n,\beta/7)\right)^2,
\end{equation*}
where the last inequality is by (\ref{eq:th:err_var_1}), (\ref{eq:th:err_var_2}), and (\ref{eq:th:err_var_3}),
and 
\begin{align}
\nonumber
\left|H'\cap \overline{[0,\widetilde{\mathrm{rad}}(H')]}\right|= & \left|H'\cap \left[-\widetilde{\mathrm{rad}}(H'),\widetilde{\mathrm{rad}}(H')\right]\right|
\nonumber
\\
\leq & c_5\log\left({1\over \beta}\log{\gamma\left(\varepsilon n,\beta/7\right) \over \varphi\left(1/16\right)}\right),
\label{eq:th:err_var_5}
\end{align}
where $c_5$ is a universal constant. The first equality is because $H'$ has the non-negative domain, and the inequality is by (\ref{eq:th:err_var_1}) and (\ref{eq:th:err_var_2}). Then, based on (\ref{eq:th:err_var_5}), with a similar idea used in the proof of Theorem~\ref{th:err_mean}, we have with probability at least $1-\frac{2\beta}{7}$,
\[\left|H\cap \overline{[0,\widetilde{\mathrm{rad}}(H')]}\right| = O\left(\frac{1}{\varepsilon}\log\left({1\over \beta}\log{\gamma\left(\varepsilon n,\beta/7\right) \over \varphi\left(1/16\right)}\right)\right).\]

The remaining analysis is almost identical to that of Theorem~\ref{th:err_mean}, with two differences. First, the $\Bar{\mathcal{T}}$ is obtained by truncating $\mathcal{T}$ into the interval $[2\sigma^2-\xi,2\sigma^2+\xi]$ for a given $\xi\geq 5\cdot \left(\gamma(\varepsilon n,\frac{\beta}{7})\right)^2+2\sigma^2$. This can also ensure $[0,\widetilde{\mathrm{rad}}(H')]\subseteq [2\sigma^2-\xi,2\sigma^2+\xi]$. Second, we replace $\mu_{\mathcal{T}}$, $\sigma^2_{\mathcal{T}}$ with $2\sigma^2$, $2\mu_{4}$ following (\ref{eq:mean_T}), (\ref{eq:sigma_T}). Finally, noting that $Z$ is non-negative, we obtain the claimed error bound.
\end{proof}

\subsection{Error Bounds for Specific Distributions}
\label{sec:var_error_specific}

\paragraph{Gaussian distributions.}
For a Gaussian $\mathcal{P}$, in addition to the properties listed in Section~\ref{sec:mean_error_specific}, we also have $\mu_4 = 3\sigma^4$. 
Plugging these properties into Theorem~\ref{th:err_var} and setting $\xi = c\sigma^2 \log(\varepsilon n)$ for some large constant $c$, we have

\begin{theorem}
\label{th:err_var_gaussian}
Given $\varepsilon$, $\beta$, for any $D\sim \mathcal{P}^n$, where $\mathcal{P}$ is a Gaussian distribution, if 
\[n> \frac{c_1}{\varepsilon}\log\log(\sigma) + \frac{c_2}{\varepsilon}\log\log(1/\sigma),\]
where $c_1,c_2$ is a universal constant, then
\[\mathrm{Err}(\mathtt{EstimateVariance},D)=O\left(\frac{\sigma^2}{\sqrt{n}} + \frac{\sigma^2}{\varepsilon n}\log\log\log( \varepsilon n)\log(\varepsilon n)\right).\]
\end{theorem}

\begin{proof}
We only need to bound the term $\left|\E\left[Z>c\sigma^2\log(\varepsilon n)\right]\right|$. For $X,X'\sim \mathcal{P}$, let $W= X-X'$, then $W\sim \mathcal{N}(\mu,2\sigma^2)$. For $t = c\sigma^2 \log(\varepsilon n)$, $c$ is large enough,
\[\E\left[Z>t\right] \leq \int_{t}^{\infty} z f_{\mathcal{T}}(z)\, \mathrm{d}z=2\int^{\infty}_{\sqrt{t}}w^2\cdot \frac{1}{2\sigma\sqrt{\pi}}\cdot\exp\left(-\frac{w^2}{4\sigma^2}\right)\, \mathrm{d}w= O\left(\frac{\sigma^2}{\varepsilon n}\sqrt{\log(\varepsilon n)}\right).\]
\end{proof}

\paragraph{Heavy-tailed distributions.} 
Now, we consider the case where $\mathcal{P}$ has a bounded $k$th-central moment $\mu_{k}$ for some $k\geq 4$. First, besides the properties mentioned in Section~\ref{sec:mean_error_specific}, we have $\mu_4 \leq \mu_k^{4/k}$.  Then, we plug  these properties into Theorem~\ref{th:err_var} and set $\xi = c\cdot \left(\varepsilon n \mu_k\right)^{2/k}$, for some $c$ large enough. The requirement on $n$ changes to 
\[n>\frac{c_1}{\varepsilon}\log\log{1 \over \varphi\left(1/16\right)}+\frac{c_2}{\varepsilon}\log\log\left(\mathrm{IQR}\right),\]
and the error bound becomes
\begin{align*}
\mathrm{Err}(\mathtt{EstimateVariance},D)
=O\left(\sqrt{\frac{\mu_4}{n}}+\frac{\mu_k^{2/k}}{(\varepsilon n)^{1-2/k}}\log\log{(\varepsilon n \mu_k)^{1/k} \over \varphi(1/16)}+\left|\E\left[Z>2\sigma^2+c\cdot\left(\varepsilon n \mu_k\right)^{2/k}\right]\right|\right).
\end{align*}

We further analyze the last term.  First, we derive a connection between $\mu_k$ and $\mu_{\mathcal{T},\frac{k}{2}}$:
\begin{lemma}
\label{lm:connection_central_moment}
For any $k\geq 2$, $\mu_{\mathcal{T},\frac{k}{2}}\leq 2^k \mu_{k}$.
\end{lemma}

\begin{proof}
\begin{align*}
\mu_{\mathcal{T},\frac{k}{2}} =& \E[(Z-2\sigma^2)^{k/2}]
\\
\leq& \E[Z^{k/2}]+2^{k/2}\sigma^k
\\
\leq& 2\E[(X-\mu)^k]+2^{k/2}\mu_k
\\
\leq& 2^k \mu_{k}.
\end{align*}
\end{proof}

Recall Lemma~\ref{lm:obs_k_moment}, and set $\xi=c\cdot\left(\varepsilon n \mu_k\right)^{2/k}$, $t = \frac{4\mu_k^{2/k}}{(\varepsilon n)^{1-2/k}}$. Since $c>8$, we have
\[\left|\E\left[Z>2\sigma^2+c\cdot\left(\varepsilon n \mu_k\right)^{2/k}\right]\right|\leq \frac{4\mu_k^{2/k}}{(\varepsilon n)^{1-2/k}}.\]

Because our analysis holds for any $k\ge 4$, Theorem \ref{th:err_var} simplifies to:
\begin{theorem}
\label{th:err_var_heavy}
Given $\varepsilon$, $\beta$, for any $D\sim\mathcal{P}^n$, if 
\[n>\frac{c_1}{\varepsilon}\log\log{1 \over \varphi\left(1/16\right)}+\frac{c_2}{\varepsilon}\log\log\left(\mathrm{IQR}\right),\]
where $c_1$, $c_2$ are universal constants, then
\[\mathrm{Err}(\mathrm{\mathtt{EstimateVariance}},D)=O\left(\sqrt{\frac{\mu_4}{n}}+\inf_{k\ge 4}\frac{\mu_k^{2/k}}{(\varepsilon n)^{1-2/k}}\log\log{(\varepsilon n \mu_k)^{1/k}\over\varphi(1/16)}\right).\]
\end{theorem}

\section{IQR Estimation}
\label{sec:scale}

Our IQR estimator is simple: We first obtain a lower bound $\overline{\mathrm{IQR}}$ with $\mathtt{EstimtaeIQRLowerBound}$. Then, we discretize $\mathbb{R}$ using bucket size $b=\overline{\mathrm{IQR}}/n$ and run $\mathtt{InfiniteDomainQuantile}$ to find $\tilde{X}_{n/4}$ and $\tilde{X}_{3n/4}$. The details are shown in Algorithm~\ref{alg:IQR}.

\begin{algorithm}
\LinesNumbered 
\label{alg:IQR}
\caption{$\mathtt{EstimateIQR}$.}
\KwIn{$D$, $\varepsilon$, $\beta$}
$\overline{\mathrm{IQR}}=\mathtt{EstimateIQRLowerBound}(D,\frac{\varepsilon}{3},\frac{\beta}{6})$\;
$\tilde{X}_{n/4} = \mathtt{InfiniteDomainQuantile}(D,n/4,\frac{\varepsilon}{3},\frac{\beta}{6})$ with $b=\overline{\mathrm{IQR}}/n$\;
$\tilde{X}_{3n/4} = \mathtt{InfiniteDomainQuantile}(D,3n/4,\frac{\varepsilon}{3},\frac{\beta}{6})$ with $b=\overline{\mathrm{IQR}}/n$\;
$\widetilde{\mathrm{IQR}} = \tilde{X}_{3n/4} - \tilde{X}_{n/4}$\;
\Return $\widetilde{\mathrm{IQR}}$\;
\end{algorithm}

The privacy of $\mathtt{EstimateScale}$ is straightforward. To analyze its utility, we introduce the following parameter $\theta(\kappa)$, which, similar to $\varphi(\beta)$, also measures how well-behaved $\mathcal{P}$ is.  While $\varphi(\beta)$ checks if $\mathcal{P}$  has a high and narrow peak, $\theta(\kappa)$ ensures that $\mathcal{P}$ has non-negligible probability mass around $F^{-1}(1/4)$ and $F^{-1}(3/4)$.  For any $\kappa\geq 0$, define the following four intervals near $F^{-1}(1/4)$ and $F^{-1}(3/4)$:
\begin{align*}
\mathcal{B}_1(\kappa) &= [F^{-1}(1/4)-\kappa,F^{-1}(1/4)],
\\
\mathcal{B}_2(\kappa) &= [F^{-1}(1/4),F^{-1}(1/4)+\kappa],
\\
\mathcal{B}_3(\kappa) &= [F^{-1}(3/4)-\kappa,F^{-1}(3/4)],
\\
\mathcal{B}_4(\kappa) &= [F^{-1}(3/4),F^{-1}(3/4)+\kappa].
\end{align*}
Then $\theta(\kappa)$ is the smallest average probability density in those four regions, i.e.,
\[\theta(\kappa) = \frac{1}{\kappa} \cdot \min_{i\in[4]}\left\{\int_{x\in\mathcal{B}_i(\kappa)}f(x)\,\mathrm{d}x\right\}.\]

Note that prior work on this problem \cite{dwork2009differential} defined something more strict, using the \textit{minimum} probability density in these regions.  Nevertheless, their analysis actually still holds under our definition of $\theta(\kappa)$. 

We first analyze the sampling error:

\begin{lemma}
\label{lm:obs_quantile_data_and_distribution}
Given $\beta, t, \kappa>0$, for any $D\sim \mathcal{P}^n$, if
\[n> \frac{2t}{\kappa\cdot \theta(\kappa)}+\frac{16\log(4/\beta)}{\left(\kappa\cdot\theta(\kappa)\right)^2},\]
then with probability at least $1-\beta$, for any $k_1\in [n/4\pm t]$ and any $k_2\in [3n/4\pm t]$, we have
\[X_{k_1}\in [F^{-1}(1/4)\pm \kappa],\]
and
\[X_{k_2}\in [F^{-1}(3/4)\pm \kappa].\]
\end{lemma}

\begin{proof}
By the condition of $n$, we have
\begin{align}
\nonumber
\begin{cases}
n> \frac{2t}{\kappa\cdot \theta(\kappa)}
\\
n> \frac{16\log(4/\beta)}{\left(\kappa\cdot\theta(\kappa)\right)^2}
\end{cases}
\Rightarrow &
\begin{cases}
\frac{1}{2} n \kappa\cdot \theta(\kappa) > t
\\
\frac{1}{2} n \kappa\cdot \theta(\kappa) > 2\sqrt{\log(4/\beta)n}
\end{cases}
\\
\Rightarrow & 
\label{eq:lm:obs_quantile_data_and_distribution_1}
n \kappa \theta(\kappa)>t+2\sqrt{\log(4/\beta)n}.
\end{align}

We define four events:
\begin{align*}
E_1 :& X_{n/4+t}> F^{-1}(1/4)+\kappa;
\\
E_2 :& X_{n/4-t}< F^{-1}(1/4)-\kappa;
\\
E_3 :& X_{3n/4+t}> F^{-1}(3/4)+\kappa;
\\
E_4 :& X_{3n/4-t}> F^{-1}(3/4)-\kappa.
\end{align*}
It suffices to show that each event happens with probability less than $\frac{\beta}{4}$.
Below we only consider $E_1$; the other 3 events are similar.
\begin{align*}
& \Pr\left[ X_{n/4+t}> F^{-1}(1/4)+\kappa\right] 
\\
=& \Pr\left[\left| D\cap \left(-\infty, F^{-1}(1/4)+\kappa\right] \right|\leq n/4+t\right] 
\\
\leq & \exp \left(-\frac{1}{4} \frac{\left(\E\left[\left| D\cap \left(-\infty, F^{-1}(1/4)+\kappa\right] \right|\right]-n/4-t\right)^2}{\E\left[\left| D\cap \left(-\infty, F^{-1}(1/4)+\kappa\right] \right|\right]}\right)
\\
\leq &  \exp \left(-\frac{1}{4} \frac{\left( n\left(1/4+ \kappa\cdot\theta(\kappa)\right)-n/4-t\right)^2}{n}\right)
\\
\leq &\frac{\beta}{4}.
\nonumber
\end{align*}
The first inequality is by Chernoff's inequality. The second one is because
\[n\left(1/4+\kappa\cdot \theta(\kappa)\right)\leq\E\left[\left| D\cap \left[-\infty, F^{-1}(1/4)+\kappa\right] \right|\right]\leq n.\]
And the last one is by (\ref{eq:lm:obs_quantile_data_and_distribution_1}).
\end{proof}

Now, we are ready to analyze the utility of $\mathtt{EstimateIQR}$.

\begin{theorem}
\label{th:err_scale}
Given $\varepsilon,\beta$, for any $D\sim \mathcal{P}^n$ and any $\alpha>0$, if 
\begin{align*}
n>& \frac{c_1}{\varepsilon}\log\log{1\over \varphi\left(1/16\right)}+\frac{c_2}{\varepsilon}\log\log\left(\mathrm{IQR}\right)+\frac{c_3}{\varepsilon}\log{1\over \beta}+\frac{c_4}{\varepsilon}\log{|\mu|+\sigma+\gamma\left(n,1/16\right) \over \varphi\left(1/16\right)}
\\
&+\frac{c_5}{\varepsilon \alpha\cdot \theta(\alpha/4)}\log{ \gamma\left(n,\beta/6\right) \over \beta\cdot \varphi\left(1/16\right)} +\frac{c_6\log(1/\beta)}{(\alpha\cdot\theta(\alpha/4))^2}+ \frac{c_7\mathrm{IQR}}{\alpha},
\end{align*}
where $c_1$, $c_2$, $c_3$, $c_4$, $c_5$, $c_6$, and $c_7$ are universal constants, then with probability at least $1-\beta$, $\mathtt{EstimateIQR}$ returns a $\widetilde{\mathrm{IQR}}$ such that $|\mathrm{IQR}-\widetilde{\mathrm{IQR}}|\leq \alpha$.
\end{theorem}

We first explain each term in the sample complexity above before presenting the proof.  The first 4 terms are the minimum requirement on the sample size, which is needed to find the right bucket size and range of $D$ so as to reduce the domain size to finite.  The last 3 terms represent the sample size-accuracy trade-off. $\tilde{O}\left({1\over \varepsilon \alpha\cdot \theta(\alpha/4)}\right)$ is the privacy term while  $\tilde{O}\left({1\over (\alpha\cdot \theta(\alpha/4))^2}\right)$ is the sampling error. The last term is the error caused by discretization. If we assume $\theta(\alpha)$ does not change much for small $\alpha$ (i.e., $f$ does not change too abruptly near $F^{-1}(1/4)$ and $F^{-1}(3/4)$), then we obtain the right convergence rate $\alpha \propto {1\over \varepsilon n} + {1 \over \sqrt{n}}$.  On the other hand, the previous work \cite{dwork2009differential} only achieves a convergence rate of $\alpha \propto {1 \over \varepsilon \log n}$.

\begin{proof}
First, by Theorem~\ref{th:err_bucket} and setting $c_1$, $c_2$, and $c_3$ sufficiently large, with probability at least $1-\frac{\beta}{6}$, we have
\begin{equation}
\label{eq:th:err_scale_1}
\frac{1}{4n} \cdot \varphi(\frac{1}{16})  \leq \bucket \leq \frac{\mathrm{IQR}}{n}.
\end{equation}

Second, by definition of $(m,\beta)$-statistical width, with probability at least $1-\frac{\beta}{6}$,
\begin{equation}
\label{eq:th:err_scale_2}
\gamma(D)\leq \gamma\left(n,\frac{\beta}{6}\right).
\end{equation}

Similar to the proof of Theorem~\ref{th:err_mean}, we further have with probability at least $1-\frac{\beta}{6}$,
\begin{equation}
\label{eq:th:err_scale_3}
\mathrm{rad}(D)\leq |\mu|+2\sigma+\gamma\left(n,\frac{\beta}{6}\right).
\end{equation}

Then, by Theorem~\ref{th:err_unbounded_quantile}, (\ref{eq:th:err_scale_1}), (\ref{eq:th:err_scale_2}), (\ref{eq:th:err_scale_3}) and setting $c_3$, $c_4$ sufficiently large, we have with probability at least $1-\frac{\beta}{3}$,
\begin{equation}
\label{eq:th:err_scale_4}
X_{n/4-t}-\bucket \leq \tilde{l}\leq X_{n/4+t}+\bucket,
\end{equation}
and
\begin{equation}
\label{eq:th:err_scale_6}
X_{3n/4-t}-\bucket \leq \tilde{r}\leq X_{3n/4+t}+\bucket,
\end{equation}
for
\begin{equation}
\label{eq:th:err_scale_7}
t=\frac{c_8}{\varepsilon}\log{ n\cdot  \gamma\left(n,\beta/6\right) \over \beta\cdot \varphi\left(1/16\right)},
\end{equation}
where $c_8$ is some constant.

Furthermore, by Lemma~\ref{lm:obs_quantile_data_and_distribution} and setting $c_5$ and $c_6$ sufficiently large, we have with probability at least $1-\frac{\beta}{6}$, for any $k_1\in [n/4\pm t]$ and $k_2\in[3n/4\pm t]$,
\begin{equation}
\label{eq:th:err_scale_8}
X_{k_1}\in \left[F^{-1}(1/4)\pm \frac{\alpha}{4}\right],
\end{equation}
\begin{equation}
\label{eq:th:err_scale_9}
X_{k_2}\in \left[F^{-1}(3/4)\pm \frac{\alpha}{4}\right].
\end{equation}

Combining (\ref{eq:th:err_scale_4}), (\ref{eq:th:err_scale_6}), (\ref{eq:th:err_scale_7}), (\ref{eq:th:err_scale_8}), and (\ref{eq:th:err_scale_9}), we have
\[|\widetilde{\mathrm{IQR}}-\mathrm{IQR}|\le \frac{\alpha}{2}+\frac{2\mathrm{IQR}}{n}.\]
By setting $c_7$ sufficiently large, we have $|\widetilde{\mathrm{IQR}}-\mathrm{IQR}|\le\alpha$.
\end{proof}

\section{Acknowledgements}
This work has been supported by HKRGC under grants 16201819, 16205420, and 16205422.  We would also like to thank Yuchao Tao for some helpful initial discussions on the problem and the anonymous reviewers who have made valuable suggestions on improving the presentation of the paper.

\bibliographystyle{alpha}
\bibliography{ref}

\newcommand{\etalchar}[1]{$^{#1}$}
\begin{thebibliography}{JKMW19}

\bibitem[AAAK21]{aden2021sample}
Ishaq Aden-Ali, Hassan Ashtiani, and Gautam Kamath.
\newblock On the sample complexity of privately learning unbounded
  high-dimensional gaussians.
\newblock In {\em Algorithmic Learning Theory}, pages 185--216. PMLR, 2021.

\bibitem[AD20]{asi2020instance}
Hilal Asi and John~C Duchi.
\newblock Instance-optimality in differential privacy via approximate inverse
  sensitivity mechanisms.
\newblock {\em Advances in neural information processing systems}, 33, 2020.

\bibitem[ADK{\etalchar{+}}19]{amin2019differentially}
Kareem Amin, Travis Dick, Alex Kulesza, Andr{\'e}s~Munoz Medina, and Sergei
  Vassilvitskii.
\newblock Differentially private covariance estimation.
\newblock In {\em NeurIPS}, pages 14190--14199, 2019.

\bibitem[AFG16]{arapinis2016sensitivity}
Myrto Arapinis, Diego Figueira, and Marco Gaboardi.
\newblock Sensitivity of counting queries.
\newblock In {\em International Colloquium on Automata, Languages, and
  Programming (ICALP)}, 2016.

\bibitem[AKMV19]{amin2019bounding}
Kareem Amin, Alex Kulesza, Andres Munoz, and Sergei Vassilvtiskii.
\newblock Bounding user contributions: A bias-variance trade-off in
  differential privacy.
\newblock In {\em International Conference on Machine Learning}, pages
  263--271. PMLR, 2019.

\bibitem[AL22]{ashtiani2021private}
Hassan Ashtiani and Christopher Liaw.
\newblock Private and polynomial time algorithms for learning gaussians and
  beyond.
\newblock In {\em Conference on Learning Theory}, pages 1075--1076. PMLR, 2022.

\bibitem[ATMR21]{andrew2019differentially}
Galen Andrew, Om~Thakkar, Brendan McMahan, and Swaroop Ramaswamy.
\newblock Differentially private learning with adaptive clipping.
\newblock {\em Advances in Neural Information Processing Systems},
  34:17455--17466, 2021.

\bibitem[BBG18]{balle2018privacy}
Borja Balle, Gilles Barthe, and Marco Gaboardi.
\newblock Privacy amplification by subsampling: Tight analyses via couplings
  and divergences.
\newblock In {\em Advances in Neural Information Processing Systems}, pages
  6277--6287, 2018.

\bibitem[BDKU20]{biswas2020coinpress}
Sourav Biswas, Yihe Dong, Gautam Kamath, and Jonathan Ullman.
\newblock Coinpress: Practical private mean and covariance estimation.
\newblock {\em Advances in Neural Information Processing Systems}, 33, 2020.

\bibitem[BDRS18]{bun2018composable}
Mark Bun, Cynthia Dwork, Guy~N Rothblum, and Thomas Steinke.
\newblock Composable and versatile privacy via truncated cdp.
\newblock In {\em Proceedings of the 50th Annual ACM SIGACT Symposium on Theory
  of Computing}, pages 74--86, 2018.

\bibitem[BGS{\etalchar{+}}21]{brown2021covariance}
Gavin Brown, Marco Gaboardi, Adam Smith, Jonathan Ullman, and Lydia
  Zakynthinou.
\newblock Covariance-aware private mean estimation without private covariance
  estimation.
\newblock {\em Advances in Neural Information Processing Systems}, 34, 2021.

\bibitem[BKN10]{beimel2010bounds}
Amos Beimel, Shiva~Prasad Kasiviswanathan, and Kobbi Nissim.
\newblock Bounds on the sample complexity for private learning and private data
  release.
\newblock In {\em Theory of Cryptography Conference}, pages 437--454. Springer,
  2010.

\bibitem[BKSW19]{bun2019private}
Mark Bun, Gautam Kamath, Thomas Steinke, and Steven~Z Wu.
\newblock Private hypothesis selection.
\newblock {\em Advances in Neural Information Processing Systems}, 32, 2019.

\bibitem[BNS13]{beimel2013private}
Amos Beimel, Kobbi Nissim, and Uri Stemmer.
\newblock Private learning and sanitization: Pure vs. approximate differential
  privacy.
\newblock In {\em Approximation, Randomization, and Combinatorial Optimization.
  Algorithms and Techniques}, pages 363--378. Springer, 2013.

\bibitem[BNS16]{bun2016Simultaneou}
Mark Bun, Kobbi Nissim, and Uri Stemme.
\newblock Simultaneous private learning of multiple concept.
\newblock In {\em Proc. Innovations in Theoretical Computer Science}, 2016.

\bibitem[BNSV15]{bun2015differentially}
Mark Bun, Kobbi Nissim, Uri Stemmer, and Salil Vadhan.
\newblock Differentially private release and learning of threshold functions.
\newblock In {\em 2015 IEEE 56th Annual Symposium on Foundations of Computer
  Science}, pages 634--649. IEEE, 2015.

\bibitem[BS16]{bun2016concentrated}
Mark Bun and Thomas Steinke.
\newblock Concentrated differential privacy: Simplifications, extensions, and
  lower bounds.
\newblock In {\em Theory of Cryptography Conference}, pages 635--658. Springer,
  2016.

\bibitem[BS19]{BunS19}
Mark Bun and Thomas Steinke.
\newblock Average-case averages: Private algorithms for smooth sensitivity and
  mean estimation.
\newblock In {\em Advances in Neural Information Processing Systems 32},
  NeurIPS '19, pages 181--191. Curran Associates, Inc., 2019.

\bibitem[CCMS19]{chan2019foundations}
TH~Hubert Chan, Kai-Min Chung, Bruce~M Maggs, and Elaine Shi.
\newblock Foundations of differentially oblivious algorithms.
\newblock In {\em Proceedings of the Thirtieth Annual ACM-SIAM Symposium on
  Discrete Algorithms}, pages 2448--2467. SIAM, 2019.

\bibitem[CSS13]{chaudhuri2013near}
Kamalika Chaudhuri, Anand~D Sarwate, and Kaushik Sinha.
\newblock A near-optimal algorithm for differentially-private principal
  components.
\newblock {\em Journal of Machine Learning Research}, 14, 2013.

\bibitem[CWZ21]{cai2019cost}
T~Tony Cai, Yichen Wang, and Linjun Zhang.
\newblock The cost of privacy: Optimal rates of convergence for parameter
  estimation with differential privacy.
\newblock {\em The Annals of Statistics}, 49(5):2825--2850, 2021.

\bibitem[De12]{de2012lower}
Anindya De.
\newblock Lower bounds in differential privacy.
\newblock In {\em Theory of cryptography conference}, pages 321--338. Springer,
  2012.

\bibitem[DFY{\etalchar{+}}22]{dong2022r2t}
Wei Dong, Juanru Fang, Ke~Yi, Yuchao Tao, and Ashwin Machanavajjhala.
\newblock R2t: Instance-optimal truncation for differentially private query
  evaluation with foreign keys.
\newblock In {\em Proceedings of the 2022 International Conference on
  Management of Data}, pages 759--772, 2022.

\bibitem[DJW13]{duchi2013local}
John~C Duchi, Michael~I Jordan, and Martin~J Wainwright.
\newblock Local privacy and statistical minimax rates.
\newblock In {\em 2013 IEEE 54th Annual Symposium on Foundations of Computer
  Science}, pages 429--438. IEEE, 2013.

\bibitem[DJW18]{duchi2018minimax}
John~C Duchi, Michael~I Jordan, and Martin~J Wainwright.
\newblock Minimax optimal procedures for locally private estimation.
\newblock {\em Journal of the American Statistical Association},
  113(521):182--201, 2018.

\bibitem[DL09]{dwork2009differential}
Cynthia Dwork and Jing Lei.
\newblock Differential privacy and robust statistics.
\newblock In {\em Proceedings of the forty-first annual ACM symposium on Theory
  of computing}, pages 371--380, 2009.

\bibitem[DLY22]{dong2022differentially}
Wei Dong, Yuting Liang, and Ke~Yi.
\newblock Differentially private covariance revisited.
\newblock {\em arXiv preprint arXiv:2205.14324}, 2022.

\bibitem[DMNS06]{dwork2006calibrating}
Cynthia Dwork, Frank McSherry, Kobbi Nissim, and Adam Smith.
\newblock Calibrating noise to sensitivity in private data analysis.
\newblock In {\em Theory of cryptography conference}, pages 265--284. Springer,
  2006.

\bibitem[DNR{\etalchar{+}}09]{dwork2009complexity}
Cynthia Dwork, Moni Naor, Omer Reingold, Guy~N Rothblum, and Salil Vadhan.
\newblock On the complexity of differentially private data release: efficient
  algorithms and hardness results.
\newblock In {\em Proceedings of the forty-first annual ACM symposium on Theory
  of computing}, pages 381--390, 2009.

\bibitem[DR14]{dwork2014algorithmic}
Cynthia Dwork and Aaron Roth.
\newblock The algorithmic foundations of differential privacy.
\newblock {\em Foundations and Trends{\textregistered} in Theoretical Computer
  Science}, 9(3--4):211--407, 2014.

\bibitem[DR19]{duchi2019lower}
John Duchi and Ryan Rogers.
\newblock Lower bounds for locally private estimation via communication
  complexity.
\newblock In {\em Conference on Learning Theory}, pages 1161--1191. PMLR, 2019.

\bibitem[DTTZ14]{dwork2014analyze}
Cynthia Dwork, Kunal Talwar, Abhradeep Thakurta, and Li~Zhang.
\newblock Analyze gauss: optimal bounds for privacy-preserving principal
  component analysis.
\newblock In {\em Proceedings of the forty-sixth annual ACM symposium on Theory
  of computing}, pages 11--20, 2014.

\bibitem[DY21]{dong21:residual}
Wei Dong and Ke~Yi.
\newblock Residual sensitivity for deferentially private multi-way joins.
\newblock In {\em Proc. ACM SIGMOD International Conference on Management of
  Data}, 2021.

\bibitem[DY22]{dong2021nearly}
Wei Dong and Ke~Yi.
\newblock A nearly instance-optimal differentially private mechanism for
  conjunctive queries.
\newblock In {\em Proceedings of the 41st ACM SIGMOD-SIGACT-SIGAI Symposium on
  Principles of Database Systems}, pages 213--225, 2022.

\bibitem[GRS19]{gaboardi2019locally}
Marco Gaboardi, Ryan Rogers, and Or~Sheffet.
\newblock Locally private mean estimation: $ z $-test and tight confidence
  intervals.
\newblock In {\em The 22nd International Conference on Artificial Intelligence
  and Statistics}, pages 2545--2554. PMLR, 2019.

\bibitem[HKM22]{hopkins2021efficient}
Samuel~B Hopkins, Gautam Kamath, and Mahbod Majid.
\newblock Efficient mean estimation with pure differential privacy via a
  sum-of-squares exponential mechanism.
\newblock In {\em Proceedings of the 54th Annual ACM SIGACT Symposium on Theory
  of Computing}, pages 1406--1417, 2022.

\bibitem[HLY21]{huang2021instance}
Ziyue Huang, Yuting Liang, and Ke~Yi.
\newblock Instance-optimal mean estimation under differential privacy.
\newblock {\em Advances in Neural Information Processing Systems}, 2021.

\bibitem[HT10]{hardt2010geometry}
Moritz Hardt and Kunal Talwar.
\newblock On the geometry of differential privacy.
\newblock In {\em Proceedings of the forty-second ACM symposium on Theory of
  computing}, pages 705--714, 2010.

\bibitem[JKMW19]{joseph2019locally}
Matthew Joseph, Janardhan Kulkarni, Jieming Mao, and Steven~Z Wu.
\newblock Locally private gaussian estimation.
\newblock {\em Advances in Neural Information Processing Systems},
  32:2984--2993, 2019.

\bibitem[JNS18]{johnson2018towards}
Noah Johnson, Joseph~P Near, and Dawn Song.
\newblock Towards practical differential privacy for sql queries.
\newblock {\em Proceedings of the VLDB Endowment}, 11(5):526--539, 2018.

\bibitem[KLSU19]{KamathLSU19}
Gautam Kamath, Jerry Li, Vikrant Singhal, and Jonathan Ullman.
\newblock Privately learning high-dimensional distributions.
\newblock In {\em Proceedings of the 32nd Annual Conference on Learning
  Theory}, COLT '19, pages 1853--1902, 2019.

\bibitem[KMS{\etalchar{+}}22]{kamath2021private}
Gautam Kamath, Argyris Mouzakis, Vikrant Singhal, Thomas Steinke, and Jonathan
  Ullman.
\newblock A private and computationally-efficient estimator for unbounded
  gaussians.
\newblock In {\em Conference on Learning Theory}, pages 544--572. PMLR, 2022.

\bibitem[KMV22]{kothari2021private}
Pravesh Kothari, Pasin Manurangsi, and Ameya Velingker.
\newblock Private robust estimation by stabilizing convex relaxations.
\newblock In {\em Conference on Learning Theory}, pages 723--777. PMLR, 2022.

\bibitem[KSSU20]{kamath2020differentially}
Gautam Kamath, Or~Sheffet, Vikrant Singhal, and Jonathan Ullman.
\newblock Differentially private algorithms for learning mixtures of separated
  gaussians.
\newblock In {\em 2020 Information Theory and Applications Workshop (ITA)},
  pages 1--62. IEEE, 2020.

\bibitem[KSU20]{kamath2020private}
Gautam Kamath, Vikrant Singhal, and Jonathan Ullman.
\newblock Private mean estimation of heavy-tailed distributions.
\newblock In {\em Conference on Learning Theory}, pages 2204--2235. PMLR, 2020.

\bibitem[KTH{\etalchar{+}}19]{kotsogiannis2019privatesql}
Ios Kotsogiannis, Yuchao Tao, Xi~He, Maryam Fanaeepour, Ashwin Machanavajjhala,
  Michael Hay, and Gerome Miklau.
\newblock Privatesql: a differentially private sql query engine.
\newblock {\em Proceedings of the VLDB Endowment}, 12(11):1371--1384, 2019.

\bibitem[KV18]{karwa2018finite}
Vishesh Karwa and Salil Vadhan.
\newblock Finite sample differentially private confidence intervals.
\newblock In {\em 9th Innovations in Theoretical Computer Science Conference
  (ITCS 2018)}. Schloss Dagstuhl-Leibniz-Zentrum fuer Informatik, 2018.

\bibitem[LKKO21]{liu2021robust}
Xiyang Liu, Weihao Kong, Sham Kakade, and Sewoong Oh.
\newblock Robust and differentially private mean estimation.
\newblock {\em Advances in Neural Information Processing Systems}, 34, 2021.

\bibitem[LKO22]{liu2021differential}
Xiyang Liu, Weihao Kong, and Sewoong Oh.
\newblock Differential privacy and robust statistics in high dimensions.
\newblock In {\em Conference on Learning Theory}, pages 1167--1246. PMLR, 2022.

\bibitem[McS09]{mcsherry2009privacy}
Frank~D McSherry.
\newblock Privacy integrated queries: an extensible platform for
  privacy-preserving data analysis.
\newblock In {\em Proceedings of the 2009 ACM SIGMOD International Conference
  on Management of data}, pages 19--30, 2009.

\bibitem[MRTZ17]{mcmahan2017learning}
H~Brendan McMahan, Daniel Ramage, Kunal Talwar, and Li~Zhang.
\newblock Learning differentially private recurrent language models.
\newblock {\em arXiv preprint arXiv:1710.06963}, 2017.

\bibitem[NH12]{narayan2012djoin}
Arjun Narayan and Andreas Haeberlen.
\newblock Djoin: Differentially private join queries over distributed
  databases.
\newblock In {\em USENIX Symposium on Operating Systems Design and
  Implementation}, pages 149--162, 2012.

\bibitem[NRS07]{nissim2007smooth}
Kobbi Nissim, Sofya Raskhodnikova, and Adam Smith.
\newblock Smooth sensitivity and sampling in private data analysis.
\newblock In {\em Proceedings of the thirty-ninth annual ACM symposium on
  Theory of computing}, pages 75--84, 2007.

\bibitem[PGM14]{proserpio2014calibrating}
Davide Proserpio, Sharon Goldberg, and Frank McSherry.
\newblock Calibrating data to sensitivity in private data analysis.
\newblock {\em Proceedings of the VLDB Endowment}, 7(8), 2014.

\bibitem[PS12]{Palamidessi2012DifferentialPF}
Catuscia Palamidessi and Marco Stronati.
\newblock Differential privacy for relational algebra: Improving the
  sensitivity bounds via constraint systems.
\newblock In {\em QAPL}, 2012.

\bibitem[PSY{\etalchar{+}}19]{pichapati2019adaclip}
Venkatadheeraj Pichapati, Ananda~Theertha Suresh, Felix~X Yu, Sashank~J Reddi,
  and Sanjiv Kumar.
\newblock Adaclip: Adaptive clipping for private sgd.
\newblock {\em arXiv preprint arXiv:1908.07643}, 2019.

\bibitem[She17]{sheffet2017differentially}
Or~Sheffet.
\newblock Differentially private ordinary least squares.
\newblock In {\em International Conference on Machine Learning}, pages
  3105--3114. PMLR, 2017.

\bibitem[Smi11]{smith2011privacy}
Adam Smith.
\newblock Privacy-preserving statistical estimation with optimal convergence
  rates.
\newblock In {\em Proceedings of the forty-third annual ACM symposium on Theory
  of computing}, pages 813--822, 2011.

\bibitem[THMR20]{tao2020computing}
Yuchao Tao, Xi~He, Ashwin Machanavajjhala, and Sudeepa Roy.
\newblock Computing local sensitivities of counting queries with joins.
\newblock In {\em Proceedings of the 2020 ACM SIGMOD International Conference
  on Management of Data}, pages 479--494, 2020.

\bibitem[TS13]{thakurta2013differentially}
Abhradeep~Guha Thakurta and Adam Smith.
\newblock Differentially private feature selection via stability arguments, and
  the robustness of the lasso.
\newblock In {\em Conference on Learning Theory}, pages 819--850. PMLR, 2013.

\bibitem[Upa18]{upadhyay2018price}
Jalaj Upadhyay.
\newblock The price of privacy for low-rank factorization.
\newblock In {\em NeurIPS}, 2018.

\bibitem[Vad17]{vadhan2017complexity}
Salil Vadhan.
\newblock The complexity of differential privacy.
\newblock In {\em Tutorials on the Foundations of Cryptography}, pages
  347--450. Springer, 2017.

\end{thebibliography}
\end{document}